\documentclass[12pt]{amsart} \usepackage{amsmath,bm}
\usepackage{amssymb} \usepackage{amsthm}
\usepackage[hidelinks]{hyperref}
\usepackage{graphicx}

\newtheorem{theorem}{Theorem}
\newtheorem{lemma}[theorem]{Lemma}

\newtheorem{corollary}[theorem]{Corollary}

\newtheorem{assumption}{Assumption}

\numberwithin{theorem}{section}
\numberwithin{equation}{section}

\theoremstyle{definition}

\newtheorem{definition}[theorem]{Definition}

\newtheorem{remark}[theorem]{Remark}

\title[Opinion Dynamics]{Opinion Dynamics Models with Memory in Coopetitive Social Networks: Analysis, Application and Simulation}

\author{Qingsong Liu}
\email{qingsongliu@wust.edu.cn}
\address{Engineering Research Center of
Metallurgical Automation and Measurement Technology, Wuhan University of
Science and Technology, Wuhan 430081, China}
\author{Li Chai}
\email{chaili@wust.edu.cn}
\address{Engineering Research Center of
Metallurgical Automation and Measurement Technology, Wuhan University of
Science and Technology, Wuhan 430081, China}
\usepackage{amstext}

\begin{document}
	\batchmode
\begin{abstract}
  In some social networks, the opinion forming is
based on its own and neighbors' (initial) opinions, whereas the evolution of the individual opinions is also influenced by the individual's past opinions in the real world. Unlike existing social network models, in this paper, a novel model of
opinion dynamics is proposed, which describes the evolution of the individuals' opinions not only depends on
its own and neighbors' current opinions, but also depends on past opinions. Memory and memoryless communication rules are
simultaneously established for the proposed opinion dynamics model. Sufficient and/or necessary conditions for the equal polarization,
consensus and neutralizability of the opinions are respectively presented
in terms of the network topological structure and the spectral analysis. We apply our model to simulate Kahneman's seminal experiments on choices in risky and riskless contexts, which fits
in with the experiment results. Simulation analysis shows that the memory capacity of the
individuals is inversely proportional to the speeds of the ultimate opinions
formational.
\end{abstract}

	\maketitle

\section{Introduction}

Analysis and control of agent-based network systems have been an active
research topic in the past decades, and the consensus problems have been
extensively studied in the literature (see,
\cite{om04tac,mz10tac,lz20tsmc,zl14auto} and the references therein).
The problem of multi-agent consensus from a graph signal processing
perspective was considered in \cite{ycz20tac}, where analytic solutions were provided for the optimal convergence rate as well as the corresponding control gains. In the theory of agent-based
network systems, the social network is one of the important and interesting
case study, the opinions of the social individuals usually reach
disagreement in the social networks \cite{pptf17tac}.
For example, in the
cooperative social networks, the disagreement of the heterogeneous belief
systems was investigated in \cite{ylwa19tac}, where it was revealed that the
disagreement behaviors of opinion dynamics were affected by the logical
interdependence structure \cite{ylwa19tac}, and for the opinion dynamics in
the antagonistic social networks, the disagreement problem under the
leader-follower hierarchical framework was studied in \cite{mmh19siam}.
In the opinion dynamics model with biased assimilation,
how individual biases influence social equilibria was reported in \cite{cqlqbs19auto}.
In the nonlinear opinion dynamics model, sufficient conditions guaranteeing asymptotic convergence of opinions were provided in \cite{am19auto}.

In the past few years, there has been an increasing interest in the study of
the opinion dynamics, and many results have been reported in the literature to
analyze the individuals' opinions on a topic evolve over time as they interact
\cite{pt17arc,yqaac19auto}. For instance, the opinion dynamics model with
bounded confidence was investigated in \cite{hk02jass}, which is called
Hegselmann-Krause (H-K) model, the consensus and polarization problem were also
addressed in \cite{hk02jass}, the extension opinion dynamics of the
H-K model with decaying confidence was considered in
\cite{mg10tac}, the evolution of opinions on a sequence of
issues was studied in \cite{jtfb15siam}, and the quasi-consensus behavior of H-K opinion dynamics was analyzed in \cite{sch17auto}. Some other well-known opinion
dynamics model (for example, DeGroot model \cite{d74jasa} and Friedkin-Johnsen (F-J)
model \cite{fj99agp}) were also studied in the literature, such as, in
strongly connected networks, a nonlinear opinion dynamics model was analyzed
in \cite{xylcs20auto}. It was shown that the stability of certain equilibria
subject to both the degree of bias and the neighbors' topology
\cite{xylcs20auto}. The multidimensional F-J model was addressed
by \cite{pptf17tac} and F-J model over issue sequences with
bounded confidence was studied in \cite{tw18auto}. To describe the stochastic evolution of opinion dynamics, a novel opinion dynamics model was proposed in \cite{bcn19auto}.

All the references mentioned above mainly take into account the evolution of
opinions in the cooperative networks. However, in some real world scenarios or
social networks, it is reasonable to assume that some individuals cooperate
with each other, the other individuals compete with each other, wherein the
positive weights among individuals implies cooperation and the negative
weights among individuals means competition, which can be described as signed
graphs. Recently, the signed graph was firstly applied to the social networks
by Altafini \cite{altafini12,altafini13tac}, it was shown that each agent can
be asymptotically converged to a value that equal size but opposite sign in
the structurally balanced networks. Afterwards, social networks with
competitive interactions have attracted extensive attentions, such as, by
using the Perron-Frobenius theorem to predict the outcomes of the opinion
forming process, which was applied to the PageRank with negative links
\cite{ag15tac}, in the time-varying network topology, the consensus and
polarization of the continuous-time opinion dynamics with hostile camps was
studied in \cite{pmc16tac}, and the opinion-forming process on the coopetitive
social networks was considered in \cite{xhc20tnse}, where the problem of how the
mass media formulates and changes public opinions was solved in
\cite{xhc20tnse}.

Notice that the literature \cite{bq15sn} reports an interesting discovery on
how human memory is encoded in social networks and the
capacity of memory mentioned in \cite{sd07sn}.
Furthermore, \cite{kahneman03ap} shows that people makes judgments and choices influenced by individual's memory.
In order to characterize the effect of human memory in the social networks, we propose a
new model of opinion dynamics on coopetitive (cooperative and competitive)
social networks, which describes the evolution of the individuals'
opinions not only depending on its own and neighbors' current opinions, but also
depending on its own and neighbors' past opinions. Memory and memoryless
communication rules are respectively established for the proposed opinion
dynamics model in the coopetitive social networks, and moreover, sufficient
and/or necessary conditions for the equal polarization, consensus and
neutralizability of the opinions are proposed on the basis of the network
topological structure and the spectral analysis. We apply the proposed model to Kahneman's seminal experiments on choices in risky and riskless contexts, which are parts of Kahneman's work in winning the Nobel Prize. We show that the model can reproduce the results of Kahneman experiments. To the best of our knowledge, few mathematical model regenerates Kahneman's experiments. According to the simulation
analysis, we find that the memory capacity of the individuals is inversely
proportional to the speeds of the ultimate opinions formation.

This paper is organized as follows. Some preliminaries and
the problem formulation are given in Section \ref{sec2}. Memory communication rules and memoryless
communication rules are established respectively in Section \ref{sec3} and Section \ref{sec4}, and sufficient and/or necessary conditions guaranteeing the equal polarization, consensus and neutralizability
of the opinions are presented. In Section \ref{sec5}, Kahneman's seminal experiments on choices in risky and riskless contexts is
examined by using our proposed opinion dynamics model. In Section \ref{sec6}, two numerical examples are
carried out to analyse the influences of the susceptibility coefficient and the memory capacity of the individuals. Section \ref{sec7} concludes
this paper.
\section{Preliminaries and Problem Formulation}\label{sec2}


Before formulating the problem of the present paper, we
first give some basic concepts and properties of signed graphs. Let
$\mathcal{G}=\left(  \mathcal{V},\mathcal{E},\mathcal{W}\right)  $ be a
weighted signed directed graph, where $\mathcal{V}=\{1,2,\cdots,N\}$ is the
set of vertices, $\mathcal{E}\subseteq$ $\mathcal{V}\times \mathcal{V}$ is the
set of edges and $\mathcal{W}=[w_{ij}]\in \mathbf{R}^{N\times N}$ is the
weighted adjacency matrix with elements $w_{ij},i,j\in \mathcal{V}$. In
addition, $w_{ij}>0$ means that individual $j$ is cooperative to individual
$i$, and $w_{ij}<0$ means that individual $j$ is competitive to individual
$i$. Similarly to \cite{altafini13tac}, we assume that $w_{ii}=0$ and
$w_{ij}w_{ji}\geq0$ for all $i,j\in \mathcal{V}$, which is called digon sign-symmetry. The Laplacian matrix
of $\mathcal{G}$ is defined as $L=\left[  l_{ij}\right]  \in \mathbf{R}^{N\times N}$, where $l_{ii}=\sum _{j=1}^{N}\left \vert w_{ij}\right \vert$, $l_{ij}=-w_{ij},i\neq j $.
Specially, the unsigned graph of the signed graph $\mathcal{G}=\left(
\mathcal{V},\mathcal{E},\mathcal{W}\right)  $ is expressed as $\overline
{\mathcal{G}}=\left(  \mathcal{V},\mathcal{E},\overline{\mathcal{W}}\right)
$, where the weighted adjacency matrix $\overline{\mathcal{W}}=[\overline
{w}_{ij}]\in \mathbf{R}^{N\times N}$ with nonnegative elements $\overline
{w}_{ij}=\left \vert w_{ij}\right \vert ,i,j\in \mathcal{V}$, and the Laplacian
matrix of $\overline{\mathcal{G}}$ is defined as $\overline{L}=\left[
\overline{l}_{ij}\right]  \in \mathbf{R}^{N\times N}$, in which,
$\overline{l}_{ii}=\sum_{j=1}^{N}\overline{w}_{ij}$,
$\overline{l}_{ij}=-\overline{w}_{ij}$, $i\neq j$. The digraph is quasi-strongly connected if it has
at least one root node, where root node has no parent and which has directed paths to all other nodes \cite{lxhb19tac,ljw19tac}. A subgraph is in-isolated if
there is no edge coming from outside to itself.

\begin{definition}
\label{definition1} A signed graph $\mathcal{G}=\left(  \mathcal{V}%
,\mathcal{E},\mathcal{W}\right)  $ is structurally balanced if the node set
can be split into two disjoint subsets $\mathcal{V}^{+}$ and $\mathcal{V}^{-}$
with the property that $\mathcal{V}^{+}\cup \mathcal{V}^{-}=\mathcal{V}$ and
$\mathcal{V}^{+}\cap \mathcal{V}^{-}=\varnothing$, $w_{ij}>0$, if
$i\in \mathcal{V}^{+},j\in \mathcal{V}^{+}$ (or $i\in \mathcal{V}^{-}%
,j\in \mathcal{V}^{-}$) and $w_{ij}<0$, if $i\in \mathcal{V}^{+},j\in
\mathcal{V}^{-}$(or $i\in \mathcal{V}^{-},j\in \mathcal{V}^{+}$).
\end{definition}

\begin{lemma}
\label{lemma1}\cite{altafini13tac} The signed graph $\mathcal{G}$ is
structurally balanced if and only if there exists a diagonal matrix
$D=\mathrm{diag}\left(  d_{1},d_{2},\cdots,d_{N}\right)  $, $d_{i}=\pm
1,i\in \mathcal{V}$, such that $\overline{\mathcal{W}}=[\overline{w}%
_{ij}]=D\mathcal{W}D$ is nonnegative matrix.
\end{lemma}

If the unsigned graph $\overline{\mathcal{G}}$ is quasi-strongly connected
(has oriented spanning tree), then there exists a nonsingular matrix
$U\in \mathbf{R}^{N\times N}$, whose first column is $\mathbf{1}_{N}%
\triangleq \left[  1,1,\ldots,1\right]  ^{\mathrm{T}}\in \mathbf{R}^{N}$, such
that \cite{rc11book}%
\begin{equation}
U^{-1}\overline{L}U=\left[
\begin{array}
[c]{ccccc}%
0 &  &  &  & \\
& \lambda_{2} & \epsilon_{2} &  & \\
&  & \ddots & \ddots & \\
&  &  & \lambda_{N-1} & \epsilon_{N-1}\\
&  &  &  & \lambda_{N}%
\end{array}
\right]  \triangleq \left[
\begin{array}
[c]{cc}%
0 & 0\\
0 & J
\end{array}
\right]
\triangleq J_{\overline{L}}, \label{j}%
\end{equation}
where $\lambda_{i},i\in \mathbf{I}[2,N]\triangleq \{2,3,\ldots,N\}$ are the
eigenvalues of the Laplacian matrix $L$, and $\epsilon_{i}\in \{0,1\},i\in
\mathbf{I}[2,N-1]$. Let $\epsilon_{N}=0$. $\overline{\mathcal{G}}$ is
quasi-strongly connected implying that $\operatorname{Re}\{ \lambda
_{i}\}>0,i\in \mathbf{I}[2,N]$.
\begin{figure}
    \centering
    \includegraphics[scale=0.5]{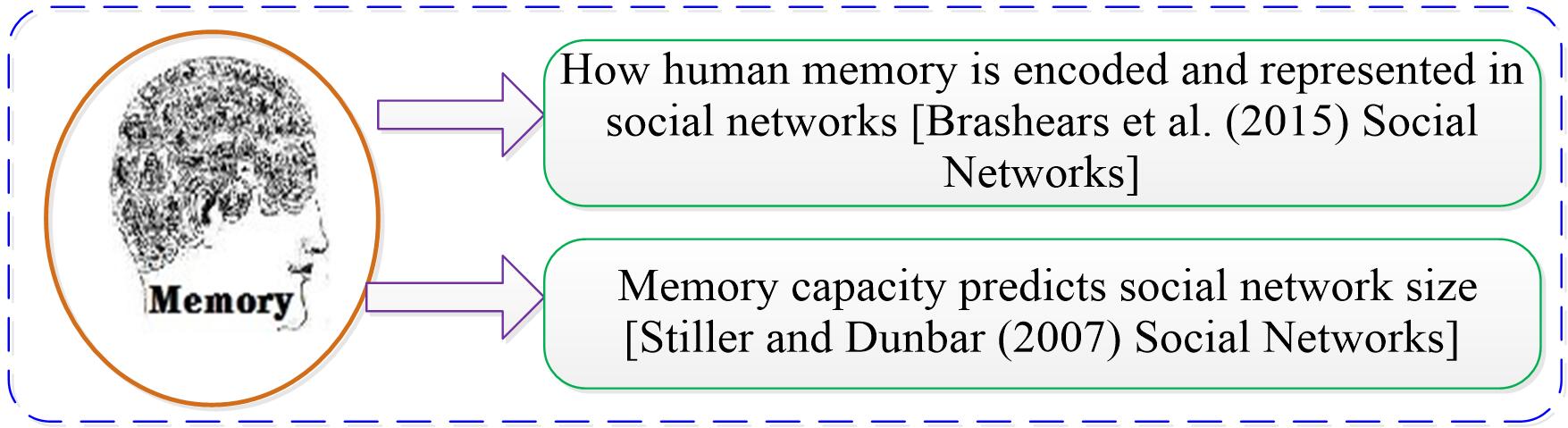}
    \caption{The main motivation of the present paper.}
    \label{fig1}
\end{figure}
\begin{lemma}
\label{lemmaadd}(Jensen Inequality \cite{gkc03book})For $P\geq0$, two scalars $\gamma_1$ and $\gamma_2$
with $\gamma_1\leq\gamma_2$, and a vector valued function $\omega:\left[  \gamma_{1},\gamma_{2}\right]  \rightarrow
\mathbf{R}^{n}$ such that the integrals in the following are well defined, then
\begin{align*}
\left(  \int_{\gamma_{1}}^{\gamma_{2}}\omega^{\mathrm{T}}\left(
\beta \right)  \mathrm{d}\beta \right)  P\left(  \int_{\gamma_{1}}^{\gamma_{2}%
}\omega \left(  \beta \right)  \mathrm{d}\beta \right)
\leq  \left(  \gamma_{2}-\gamma_{1}\right)  \left(  \int_{\gamma_{1}}%
^{\gamma_{2}}\omega^{\mathrm{T}}\left(  \beta \right)  P\omega \left(
\beta \right)  \mathrm{d}\beta \right)  .
\end{align*}
\end{lemma}

Motivated by the problem that how human memory is encoded and represented in social networks
\cite{bq15sn} and \cite{kahneman03ap} shows that people makes judgments and choices influenced by individual's memory (see Fig. \ref{fig1} for brevity). In this paper, to characterize
the effect of human memory in the social networks $\mathcal{G}$, we propose a
new opinion dynamics model as follows,%
\begin{equation}
\dot{x}_{i}\left(  t\right)  =Ax_{i}\left(  t\right)  +\int_{-h}^{0}\sigma
\left(  \theta \right) \left(I_n-A\right)u_{i}\left(  t+\theta \right)  \mathrm{d}%
\theta, \label{sys0}%
\end{equation}
where $x_{i}\left(  t\right)  \in \mathbf{R}^{n}$, $i\in \mathbf{I}[1,N]$, denotes the opinion of
individual $i$ on some topics, $u_{i}\left(  t\right)  \in \mathbf{R}^{n}$
is the communication rule of individual $i$, $h>0$ denotes the memory
capacity of the individuals, and $\int_{-h}^{0}\sigma\left(  \theta \right)  \mathrm{d}%
\theta=1$. The matrix
$A=[A_{ij}]\in \mathbf{R}^{n\times n}$ describes the ability of cognizance on their
topic opinions, which is formed by some endogenous factors and exogenous
conditions, for instance, personal intelligence, education level and the
social experience. For example, $A_{12}=1$ means that individual
$i$ has very strong cognitive ability on topic 2, and $A_{13}=-1$ implies that
individual $i$ has very weak cognitive ability on topic 3. In order to make the ideas in the present
paper easy to follow, we consider the following simplified model of (\ref{sys0}), $i\in \mathbf{I}[1,N]$,
\begin{equation}
\dot{x}_{i}\left(  t\right)  =Ax_{i}\left(  t\right)  +\sigma\Lambda u_{i}\left(
t\right)  +(1-\sigma)\Lambda u_{i}\left(  t-h\right),
\label{sys}%
\end{equation}
which is seen as the opinion update mechanism with memory, where $\Lambda=I_{n}-A$. We consider both types of rules, one with memory and the other without memory. The communication rule with memory is given by
\begin{align}
u_{i}\left(  t\right)  =&\digamma \sum\limits_{j=1}^{N}\left \vert w_{ij}%
\right \vert \left( \mathrm{sgn}\left(  w_{ij}\right)
x_{j}\left(  t\right) -  x_{i}\left(  t\right)\right)\nonumber \\ &+\mathcal{H}\left( v_{i}\left(
t-\theta \right)  \right)  ,\theta \geq0, \label{pro1}%
\end{align}
which is seen as the communication rule (process) involved the memory capacity, where $v_{i}(t)$ is defined later in (\ref{v1}) and the memoryless communication rule of individual $i$ is given by%
\begin{equation}
u_{i}\left(  t\right)  =\digamma \sum\limits_{j=1}^{N}\left \vert w_{ij}%
\right \vert \left(  \mathrm{sgn}\left(  w_{ij}\right)
x_{j}\left(  t\right)  - x_{i}\left(  t\right) \right)  , \label{pro2}%
\end{equation}in which, $\digamma \in \mathbf{R}^{n\times n}$ is an `opinion adjustment
matrix', $\mathcal{H}$ is a linear map and $\mathrm{sgn}\left(  \cdot \right)$ denotes the signum function. From the opinion dynamics point of view model (\ref{sys}) can be explained as follows. The matrix $\sigma\Lambda\in \mathbf{R}^{n\times n}$ denotes
the susceptibility of individual $i$ to the current interpersonal influence
and the matrix $(1-\sigma)\Lambda\in \mathbf{R}^{n\times n}$ denotes the
susceptibility of individual $i$ to the past interpersonal influence, and the
susceptibility coefficient $\sigma$ satisfies $0\leq \sigma\leq1$. We mention that $\Lambda=I_{n}-A$ is to be such that the
opinion of individual $i$ on topic $p$ is in accord with his initial
opinion on topic $p$ before the individual communicates with his neighbors. In this paper, we let $x_{ip}\left(  t\right)
,p\in \mathbf{I}[1,n]$ denote the opinion on topic $p$ of individual $i$.
Specially, if $x_{ip}>0$, we say that individual $i$ support topic $p$, if
$x_{ip}<0$, we say that individual $i$ reject topic $p$, and we say that
individual $i$ remain neutral on topic $p$ if $x_{ip}=0$. In this paper, we not only consider the opinion forming based on the update rule consisting of (\ref{sys}) and (\ref{pro1}), but also consider the opinion forming based on the update rule consisting of (\ref{sys}) and (\ref{pro2}).

Notice that, a
$q$-voter model with memory was studied in \cite{js18pa}, where the social network
represented by a complete graph and the model agents are described by a single
binary variable. However, in this paper, the social network represented by a
signed graph and the model agents are described by the multi-variable. A
stochastic model and a hybrid model in opinion dynamics with collective memory
were studied in \cite{bck20pa} and \cite{mmpz20csl}, respectively. Nevertheless, the memory that
appears in the models, being a memory of past relations (network connections) \cite{bck20pa,mmpz20csl}, which is
completely different from that the memory of past opinions in the present
paper.
\begin{definition}
\label{definition2}The opinions in the social networks are equal
polarization if there exists $x^{\ast}\left(  t\right)  $, such that
$\lim_{t\rightarrow \infty}x_{i}\left(  t\right)  $ $=x^{\ast}\left(  t\right)
,i\in \mathcal{V}^{+}$ and $\lim_{t\rightarrow \infty}x_{j}\left(  t\right)
=-x^{\ast}\left(  t\right)  ,j\in \mathcal{V}^{-}$, and moreover, if
$\mathcal{V}^{+}=\varnothing$\ or $\mathcal{V}^{-}=\varnothing$, then equal
polarization reduces to consensus. Specially, the opinions are neutralizable
if $\lim_{t\rightarrow \infty}x_{i}\left(  t\right)  =0,i\in \mathcal{V}.$
\end{definition}
\begin{remark}
\label{remark1}Notice
that the opinion is nonstationary equal polarization/consensus in Definition
\ref{definition2}. However, the
opinion is stationary equal polarization/consensus when $x^{\ast}\left(
t\right)  =c$, where $c$ is a constant, which is a special case of Definition
\ref{definition2}, and also is the case of Definition 1 in \cite{altafini13tac}.
\end{remark}

Before closing this subsection, we give the following lemma to show the
properties between $L$ and $\overline{L}$, whose proof can be obtained with the help of \cite{zc17ijrnc}.

\begin{lemma}
\label{lemma2}The eigenvalues of Laplacian matrix $L$ is the same as that of
$\overline{L}$ if the signed graph $\mathcal{G}$ is structurally balanced.
Moreover, the Jordan canonical form of Laplacian matrix $L$ is the same as
(\ref{j}) if the structurally balanced signed graph $\mathcal{G}$ is
quasi-strongly connected.
\end{lemma}

\section{Communication Rules with Memory}\label{sec3}

The main purpose of this paper is to analyze the evolution of opinions in social networks by our proposed model and reproduce the results of
Kahneman experiments. In order to clearly present the theoretical results and interpretations as social phenomena,
all proofs are moved to the Appendix.

In what follows, we will establish the memory communication rule for the
opinion dynamics model (\ref{sys}) and analyze the evolution of opinions in
the social networks. Define a new state vector as%
\begin{equation}
z_{i}\left(  t\right)  =x_{i}\left(  t\right)  +(1-\sigma)\int_{t-h}%
^{t}\mathrm{e}^{A\left(  t-h-s\right)  }\Lambda u_{i}\left(  s\right)  \mathrm{d}%
s, \label{z}%
\end{equation}
where $i\in \mathbf{I}[1,N]$. It then follows from (\ref{sys}) and (\ref{z}) that%
\[
\dot{z}_{i}\left(  t\right) =Az_{i}\left(  t\right)  +\mathcal{A}u_{i}\left(
t\right)  ,i\in \mathbf{I}[1,N],
\]
where
\begin{equation}
\mathcal{A}=\sigma\Lambda+\left(  1-\sigma\right)  \mathrm{e}^{-Ah}\Lambda.
\label{b}%
\end{equation}
Then the opinion communication of each individual obeys the following rules,%
\begin{align}
u_{i}\left(  t\right)  = &  F\sum_{j\in \mathcal{N}_{i}}\left \vert
w_{ij}\right \vert \left(  \mathrm{sgn}\left(  w_{ij}\right)  z_{j}\left(
t\right)  -z_{i}\left(  t\right)  \right)\label{u1}
\end{align}where $i\in \mathbf{I}[1,N]$ and $F$ is referred to as an `opinion adjustment matrix'. It yields from
(\ref{u1}) that%
\begin{equation}
u_{i}\left(  t\right)  =F\left(  \chi_{i}\left(  t\right)  +(1-\sigma)\int_{t-h}^{t}\mathrm{e}^{A\left(  t-h-s\right)  }\Lambda v_{i}\left(  s\right)
\mathrm{d}s\right)  , \label{u2}%
\end{equation}
where $i\in \mathbf{I}[1,N]$,
\begin{align}
\chi_{i}\left(  t\right)  =\sum\limits_{j=1}^{N}\left \vert w_{ij}%
\right \vert \left( \mathrm{sgn}\left(  w_{ij}\right)
x_{j}\left(  t\right) -x_{i}\left(  t\right)  \right)  =-\sum\limits_{j=1}^{N}l_{ij}x_{j}\left(
t\right)  , \label{x1}%
\end{align}
and%
\begin{align}
v_{i}\left(  t\right)  =\sum\limits_{j=1}^{N}\left \vert w_{ij}\right \vert
\left( \mathrm{sgn}\left(  w_{ij}\right)
u_{j}\left(  t\right) -u_{i}\left(  t\right)  \right)  =-\sum\limits_{j=1}^{N}l_{ij}u_{j}\left(
t\right) . \label{v1}%
\end{align}
For simplicity, we define the set of matrices as
\begin{align*}
\mathcal{S} = &\left \{  S:\lambda(S)\subseteq \mathbf{C}_{\leq0},s\text{ is
semi-simple,}\right.  \\
& \left.  \forall s\in \lambda(S)\cap \iota \mathbf{R\neq \varnothing}\right \}
\mathbf{,}%
\end{align*}
where $\lambda(S)$ denotes the spectrum of a square matrix $S$, $\mathbf{C}%
_{\leq0}$ denotes the set of negative complex numbers and 0, a semi-simple eigenvalue possesses equal algebraic and geometric multiplicities,
$\iota^{2}=-1$ and $\iota\mathbf{R}$ denotes the imaginary axis.
Now, we have the following new theorem.

\begin{theorem}
\label{theorem1}Consider the opinion dynamic model (\ref{sys}) in the coopetitive social network
$\mathcal{G}$. The individuals' opinions are equal polarization under the communication rule (\ref{u1})
if the following conditions are satisfied: $1)$ $\mathcal{G}$ is structurally
balanced and quasi-strongly connected; $2)$ $A\in \mathcal{S}$; $3)$ There exists an `opinion adjustment
matrix' $F$ such that
$A-\lambda_{i}\mathcal{A}F,i\in \mathbf{I}[2,N]$ are Hurwitz.
\end{theorem}

Regarding the conditions in Theorem \ref{theorem1}, we give some discussions
and social interpretations as follows. In condition 1), the structurally balanced property means that the social community divides into two
hostile camps, the positive influence $w_{ij}>0$ implies that the individuals
come from the same camp and cooperate with each other in the social ties,
whereas the negative influence $w_{ij}<0$ denotes the individuals compete with
each other in the social ties. In addition, quasi-strongly connectivity implies that there is at least one individual who
can transmit directly or indirectly his/her opinions to the remaining
individuals. Condition 2) not only emphasizes that the individuals' dynamical
properties are important, but also eliminates the case that opinions are
neutralizable and unbounded. Finally, condition 3) guarantees that the
individual is open to interpersonal influence.

The social interpretations of Theorem \ref{theorem1} can be summarized as
below. In the social community with two hostile camps, if there is at least one
direct/indirect transmission line from the one individuals to the remaining
individuals and the opinions of individuals influence each other. Then some
individuals support a topic and some one reject a topic under the memory
communication rule. According to Theorem \ref{theorem1}, we have the following new
corollary whose proof is similar to that of Theorem \ref{theorem1}, thus is
omitted.

\begin{corollary}
\label{corollary1}For opinion dynamics model (\ref{sys}) in the cooperative social network $\overline{\mathcal{G}}$, the individuals' opinions achieve
consensus under the communication rule (\ref{u1}) if $\overline
{\mathcal{G}}$ is quasi-strongly connected, $A\in \mathcal{S}$ and
There exists an `opinion adjustment
matrix' $F$ such that
$A-\lambda_{i}\mathcal{A}F,i\in \mathbf{I}[2,N]$ are Hurwitz.
\end{corollary}

Corollary \ref{corollary1} means that, under the memory communication rules,
all individuals support or reject a topic in the cooperative social
community if there is at least one direct/indirect transmission line from the one
individuals to the remaining individuals and the opinions of individuals
are influenced by others. Necessary and sufficient conditions
guaranteeing the neutralizability of the opinions are given by the following new theorem.
\begin{theorem}
\label{theorem3}For opinion dynamic model (\ref{sys}) in coopetitive social network $\mathcal{G}$. The individuals' opinions are neutralizable under the communication rule (\ref{u1}) if
and only if $\mathcal{G}$ does not involve an
in-isolated structurally balanced subgraph and $(A,\mathcal{A})$ is controllable.
\end{theorem}

The meaning of Theorem \ref{theorem3} is that, under the memory
communication rule, the individuals' opinions are neutralizable in the both
cooperative and competitive if the opinions of individuals influence each
other and the social community in the absence of two hostile camps.

\section{Memoryless Communication Rules}\label{sec4}

In order to establish the memoryless communication rules, we first let the
opinion adjustment matrix $F$ be parameterized as $F=F\left(  \gamma \right)
:\mathbf{R}^{+}\rightarrow \mathbf{R}^{n\times n}$ \cite{zld12auto} and such that
\begin{equation}
\lim_{\gamma \downarrow0}\frac{1}{\gamma}\left \Vert F\left(  \gamma \right)
\right \Vert <\infty. \label{eq2}%
\end{equation}
Hence, the communication rules $u_{i}\left(  t\right)  $ is \textquotedblleft
of order $1$\textquotedblright \ with respect to $\gamma$. As a result, we can
get that the second term $F\sum_{j=1}^{N}l_{ij}\int_{t-h}%
^{t}\mathrm{e}^{A\left(  t-h-s\right)  }\Lambda u_{j}\left(  s\right)  \mathrm{d}s$
in (\ref{u2}) are at least \textquotedblleft of order 2\textquotedblright%
\ with respect to $\gamma$, which means that the term $F\sum_{j=1}^{N}l_{ij}\int_{t-h}^{t}\mathrm{e}^{A\left(  t-h-s\right)  }%
\Lambda u_{j}\left(  s\right)  \mathrm{d}s$ is dominated by the term $F\chi
_{i}\left(  t\right)  $ in (\ref{u2}), and thus might be safely neglected in
$u_{i}\left(  t\right)  $ when $\gamma$ is sufficiently small. Hence, the memoryless communication rules can be represented as
\begin{equation}
u_{i}\left(  t\right)  =F\left(  \gamma \right) \chi_{i}\left(  t\right)  . \label{u3}%
\end{equation}
The following assumption guarantees that
there exists a opinion adjustment matrix $F\left(  \gamma \right)  $ satisfying
(\ref{eq2}).

\begin{assumption}
\label{assumption1}The matrix pair $\left(  A,\mathcal{A}\right)  $ is
controllable, and all the eigenvalues of $A$ are on the imaginary axis.
\end{assumption}

If Assumption \ref{assumption1} is satisfied, the following parametric
algebraic Riccati equation (ARE)
\begin{equation}
A^{\mathrm{T}}P+PA-P\mathcal{A}\mathcal{A}^{\mathrm{T}}P=-\gamma P,
\label{eq3}%
\end{equation}
exists a unique positive definite solution $P(\gamma), \forall\gamma>0$, and $F=-\mathcal{A}^{\mathrm{T}}P\left(  \gamma \right)  $
satisfy (\ref{eq2}).

\begin{theorem}
\label{theorem2}Consider opinion dynamics model (\ref{sys}) with Assumption \ref{assumption1} in the coopetitive
social networks $\mathcal{G}$. Let $P\left(\gamma \right)$ be the unique solution of (\ref{eq3}) and define $F=-\mathcal{A}^{\mathrm{T}}P\left(  \gamma \right)  $. Then, for any $h>0$ and $\varrho \geq \max_{i\in \mathbf{I}%
[2,N]}\{1/\operatorname{Re}\{ \lambda_{i}\} \}$, there exists a number
$\gamma^{\ast}=\gamma^{\ast}\left(  \varrho,h,\{ \lambda_{i}\}_{i=2}%
^{N}\right)  $ such that the individuals' opinions are equal polarization by the memoryless communication rules%
\begin{equation}
u_{i}\left(  t\right)  =\varrho F\left(  \gamma \right)\chi_{i}\left(  t\right)  ,\gamma
\in \left(  0,\gamma^{\ast}\right)  , \label{u4}%
\end{equation}
if the signed graph $\mathcal{G}$ is structurally balanced and is
quasi-strongly connected.
\end{theorem}

This new theorem implies that, there exist possibilities of opinion
adjustment such that some individuals support a topic and others reject a
topic under memoryless communication rules if the social community can be
divided into two hostile camps. With Theorem \ref{theorem2}, we
have a new corollary as follows.

\begin{corollary}
\label{corollary3}Consider opinion dynamics model (\ref{sys}) with Assumption \ref{assumption1} in the cooperative social network $\overline{\mathcal{G}}$. Then, for
any $h>0$ and $\varrho \geq \max_{i\in \mathbf{I}[2,N]}\{1/\operatorname{Re}\{
\lambda_{i}\} \}$, there exists a number $\gamma^{\ast}=\gamma^{\ast}\left(
\varrho,h,\{ \lambda_{i}\}_{i=2}^{N}\right)  $, such that the individuals'
opinions achieve consensus for all $\gamma
\in \left(  0,\gamma^{\ast}\right)  $ by the memoryless communication rule
(\ref{u4}) if $\overline{\mathcal{G}}$ is quasi-strongly connected.
\end{corollary}

Corollary \ref{corollary3} means that, in the cooperative social community,
there exist possibilities of opinion adjustment such that all individuals
support or reject a topic under the memoryless communication rules. By Theorem
\ref{theorem2}, necessary and sufficient conditions guaranteeing the
neutralizability of the opinions for the proposed opinion dynamics model
(\ref{sys}) are given as follows.
\begin{theorem}
\label{theorem4}Consider opinion dynamics model (\ref{sys}) with Assumption \ref{assumption1} in coopetitive social networks $\mathcal{G}$. Let $P\left(  \gamma \right)  $ be unique solution of (\ref{eq3}). Then, for any $h>0$ and $\varrho \geq \max
_{i\in \mathbf{I}[2,N]}\{1/\operatorname{Re}\{ \lambda_{i}\} \}$, there exists
a number $\gamma^{\ast}=\gamma^{\ast}\left(  \varrho,h,\{ \lambda_{i}%
\}_{i=2}^{N}\right)  $, such that the individuals' opinions are neutralizable for all $\gamma \in \left(  0,\gamma^{\ast
}\right)  $ by the memoryless communication rule (\ref{u4}) if and only if
$\mathcal{G}$ does not involve an in-isolated
structurally balanced subgraph.
\end{theorem}

Theorem \ref{theorem4} means that there exist possibilities
of opinion adjustment so that individual opinions are neutralizable
under the memoryless communication rules if the social community in the
absence of two hostile camps.
\begin{remark}
\label{remark2}Compared with \cite{lxhb19tac}, this section possesses some
differences and novelty. On the one hand, we would like to emphasize that the opinion
dynamics model (\ref{sys}) involves memory $h$, and the model in this paper is more complex. However, the opinion
dynamics model in absence of the memory in \cite{lxhb19tac}. On the other hand, in this paper, the opinion adjustment
matrix $F\left(  \gamma \right)  $ of the memoryless communication rule (process) is a
variable with respect to $\gamma$, which is easily adjusted in
the real world. 
\end{remark}
\begin{remark}
\label{remark3}We will give a method to estimate the upper
bound of $\gamma^{\ast}$ as follows. According to the proof of Theorem
\ref{theorem2}, we have
\[
\dot{\eta}\left(  t\right)  =A\eta \left(  t\right)  +\sigma \Lambda \digamma
\eta \left(  t\right)  +(1-\sigma)\Lambda \digamma \eta \left(  t-h\right)  ,
\]
is asymptotically stable, where $\digamma=-\varrho \mathcal{A}^{\mathrm{T}%
}P\left(  \gamma \right)  $. The stability of the above equation is completely
determined by the rightmost roots of its characteristic equation%
\[
c\left(  s,\gamma \right)  =\det \left(  sI_{n}-A-\sigma \Lambda \digamma
-(1-\sigma)\Lambda \digamma \mathrm{e}^{-sh}\right)  .
\]
Let%
\[
\lambda_{\max}\left(  \gamma \right)  =\max \{ \operatorname{Re}\{s\}:c\left(
s,\gamma \right)  =0\},
\]
which can be efficiently computed by the software package DDE-BIFTOOL
\cite{els01}. Then the (35) is asymptotically
stable if and only if $\lambda_{\max}\left(  \gamma \right)  <0$. The maximal
$\gamma$ such that $\lambda_{\max}\left(  \gamma \right)  =0$ is defined as
$\gamma^{\ast}$.
\end{remark}

\section{Application to Kahneman's Experiments}\label{sec5}

In this section, we will use our proposed model to revisit Kahneman's seminal experiments on choices in risky and riskless contexts
\cite{kt84ap}. To totally appreciate and understand the results, we present a brief
overview of the famous experiments. Assume that a city is preparing for an outbreak of an
unusual disease (such as, COVID-$19$), and COVID-$19$ is expected to kill
$600$ people. Two alternative programs to combat COVID-$19$ have been proposed, wherein the
\textquotedblleft lives saved version" (LSV) of the scientific estimates of the
consequences for the proposed programs are given as follows.

\begin{itemize}
\item A1: If Program $1$ is adopted, $200$ people will be saved.

\item B1: If Program $2$ is adopted, the probability of $600$ people being saved is
$1/3$, and the probability of no one being saved is $2/3$.
\end{itemize}

The ``lives lost version" (LLV) of the scientific estimates of the consequences for
the proposed programs are stated as below, which is undistinguishable in real
terms from the ``LSV" \cite{kt84ap}.

\begin{itemize}
\item A2: If Program $1$ is adopted, $400$ people will die.

\item B2: If Program $2$ is adopted, the probability that no one will die is $1/3$,
and the probability that $600$ people will die is $2/3$.
\end{itemize}

The percentage who chose each option is given in Table
1 after each program \cite{kt84ap}.
It is interesting to see from Table 1 that the consequences of the
\textquotedblleft LSV" is completely different from the
consequences of the \textquotedblleft LLV" even A1 and B1 are
undistinguishable from A2 and B2, respectively.
\begin{figure}
\centering
\includegraphics[scale=0.5]{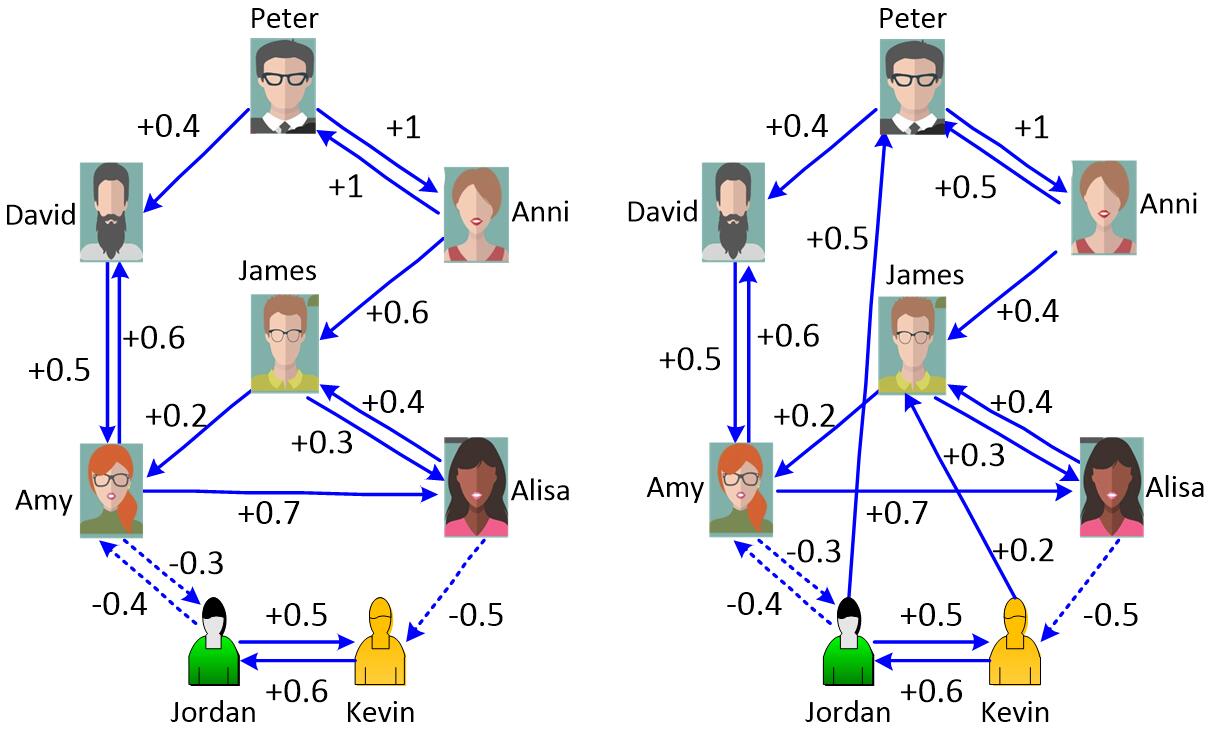}
\caption{Left: Structurally balanced community. Right: Structurally unbalanced
community.}%
\label{fig2}%
\end{figure}
\begin{figure}
\centering
\includegraphics[scale=0.65]{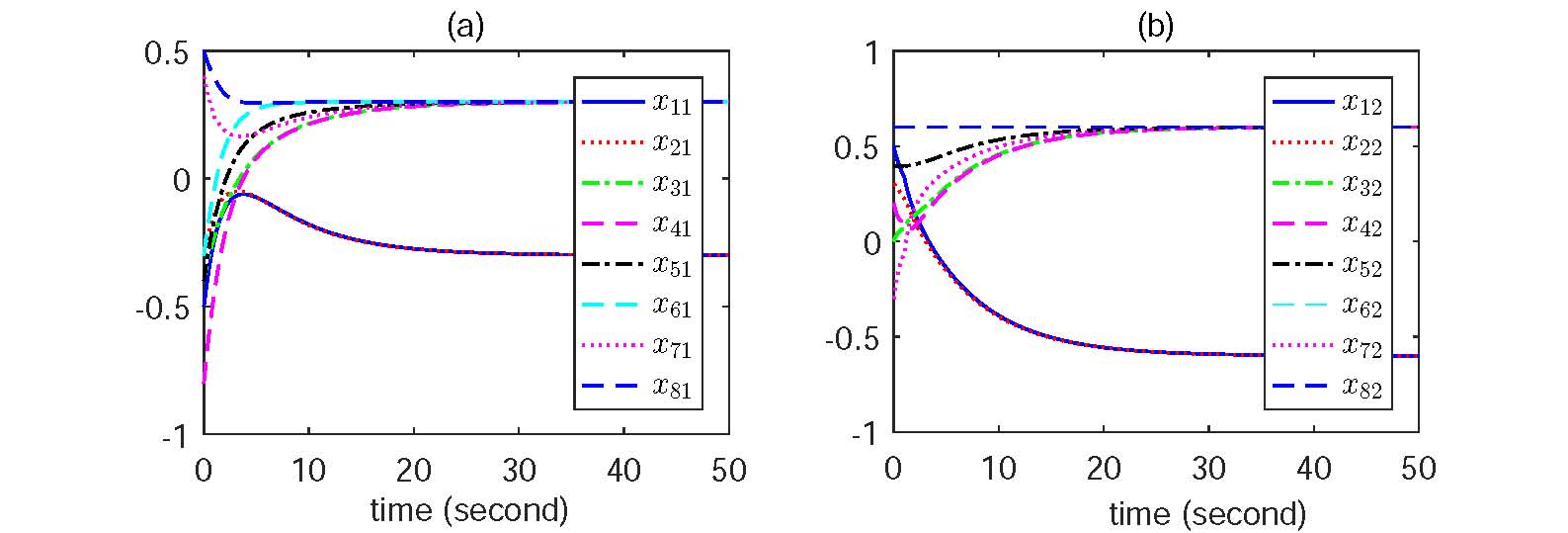}
\caption{Equal polarization of opinion dynamics model consisting of
(\ref{sys}) and (\ref{u1}) in the coopetitive community.}%
\label{fig12}%
\end{figure}
\begin{table}
\centering
Table 1: The results of the experiments for different versions
\label{table1}%
\centering
\begin{tabular}
[c]{c|c|c}\hline \hline
~~~~Different Versions~~&\,\,\,\,LSV~~~~~~~~&\,\,\,\,LLV~~~~\\ \cline{1-3}%
~~~~Different Programs~~ & \,\,A1\,\,\,\,B1 & ~A2\,\,\,\,B2~\\ \hline
~~~~Percentages~~& $~~72\%$\,\,\,\,$28\%$ & $~~22\%$\,\,\,\,$78\%$~\\ \hline
\end{tabular}
\end{table}

In the following, Kahneman's seminal experiments are simulated by the
proposed opinion dynamics model with the memory communication rules. Let the structurally balanced social community network be given in Fig. \ref{fig2} (Left), and Jordan, Kevin, Alisa, Amy, James, Anni, David, Peter are assigned labels $1,2,3,4,5,6,7$ and $8$,
respectively. In this particular case, cooperative link
means that the individuals come from government officials or well-known
doctors, whereas competitive link means that one comes from government
officials and another one comes from well-known doctors, and memory
information means that the past choices (opinions) of the government officials
and well-known doctors in risky and riskless contexts. Let $h=1$, $\sigma=0.6$, and
$A$ is given by
\[
A=\left[
\begin{array}
[c]{cc}%
-0.5 & 0.25\\
0 & 0
\end{array}
\right]  .
\]
For simplicity, without loss of generality, we explain the meaning of matrix
$A$ by A1 and B1 of \textquotedblleft LSV\textquotedblright as follows. $A_{11}=-0.5$ means that individual $i$ lacks cognitive ability on A1, $A_{12}=0.25$ means that individual $i$ possesses cognitive ability on B1 before the individual communicates with his neighbors. The corresponding
matrices $\mathcal{A}$ and $F$ are respectively computed as%
\begin{align*}
\mathcal{A}=\left[
\begin{array}
[c]{cc}%
2.4731 & -0.7365\\
0 & 1
\end{array}
\right], F=\left[
\begin{array}
[c]{cc}%
-0.0615 & -0.4059\\
0 & -1.1564
\end{array}
\right],
\end{align*}
where $F$ is such that $\lambda \left(  A-1.643\mathcal{A}F\right)
=\{-0.75,-1.9\}$.

Let the initial opinions for the six individuals be $x_{1}\left(  0\right)
=\left[  -0.5,0.5\right]  ^{\mathrm{T}}$, $x_{2}\left(  0\right)  =\left[
-0.3,0.3\right]  ^{\mathrm{T}}$, $x_{3}\left(  0\right)  =\left[
-0.4,0\right]  ^{\mathrm{T}}$, $x_{4}\left(  0\right)  =\left[
-0.8,0.2\right]  ^{\mathrm{T}}$, $x_{5}\left(  0\right)  =\left[
-0.4,0.4\right]  ^{\mathrm{T}}$, $x_{6}\left(  0\right)  =\left[
-0.3,0.6\right]  ^{\mathrm{T}}$, $x_{7}\left(  0\right)  =\left[
0.4,-0.3\right]  ^{\mathrm{T}}$ and $x_{8}\left(  0\right)  =\left[
0.5,0.6\right]  ^{\mathrm{T}}$. The simulation results are recorded in Fig. \ref{fig12},
where $x_{i1}$ denotes the opinion of individual $i$ on A1 and $x_{i2}$ denotes the opinion of individual $i$ on B2.

Clearly, according to the left of Fig. \ref{fig12}, six individuals support A1, where the percentage who chose A1 is $\frac{6}{8}=75\%$, which almost fits with the percentage $72\%$ who chose A1 in Table 1. In addition, two individuals reject A1, namely, two individuals support B1, where the percentage who chose B1 is $\frac{2}{8}=25\%$, which almost fits with the percentage $28\%$ who chose B1 in Table 1.

Similarly, in terms of the right of Fig. \ref{fig12}, six individuals support B2, where the percentage who chose B2 is $\frac{6}{8}=75\%$, which almost fits with the percentage $78\%$ who chose B2 in Table 1. In addition, two individuals reject B2, namely, two individuals support A2, where the percentage who chose A2 is $\frac{2}{8}=25\%$, which almost fits with the percentage $22\%$ who chose A2 in Table 1. Therefore, the simulation results in Fig. \ref{fig12} nearly fit in
Kahneman experiment.

\begin{figure}
\centering
\includegraphics[scale=0.65]{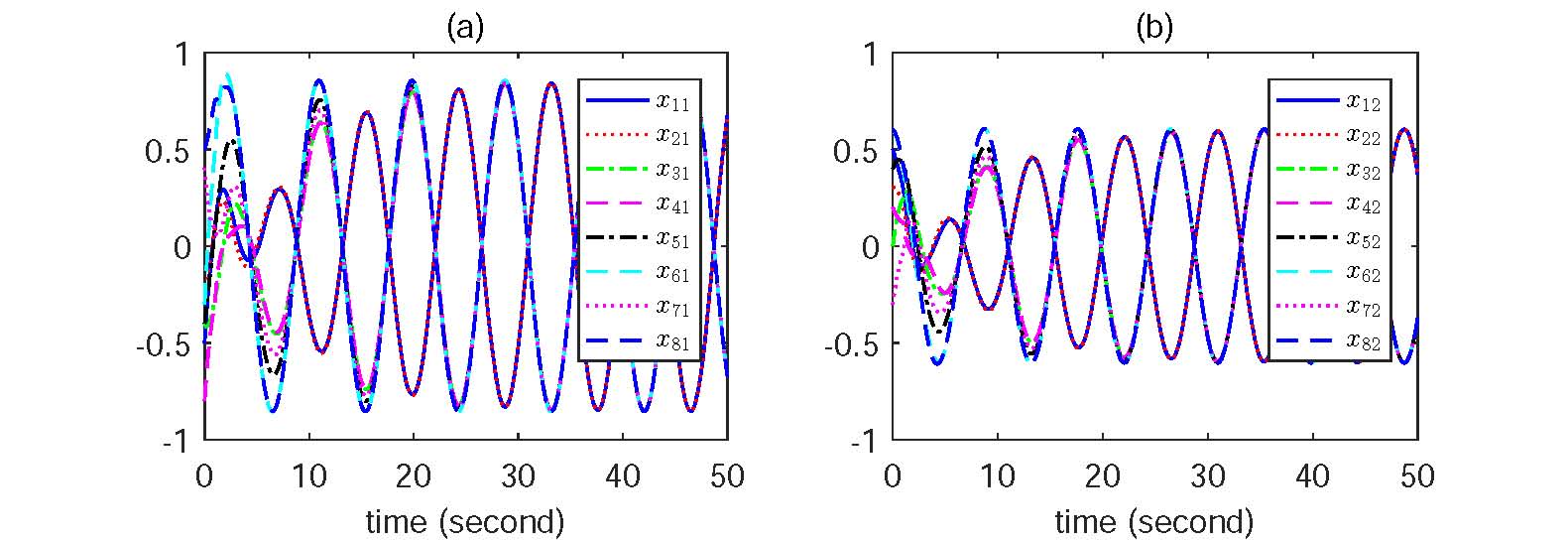}
\caption{Equal polarization of opinion dynamics model consisting of
(\ref{sys}) and (\ref{u4}) in the coopetitive community.}%
\label{fig15}%
\end{figure}
\begin{figure}
\centering
\includegraphics[scale=0.65]{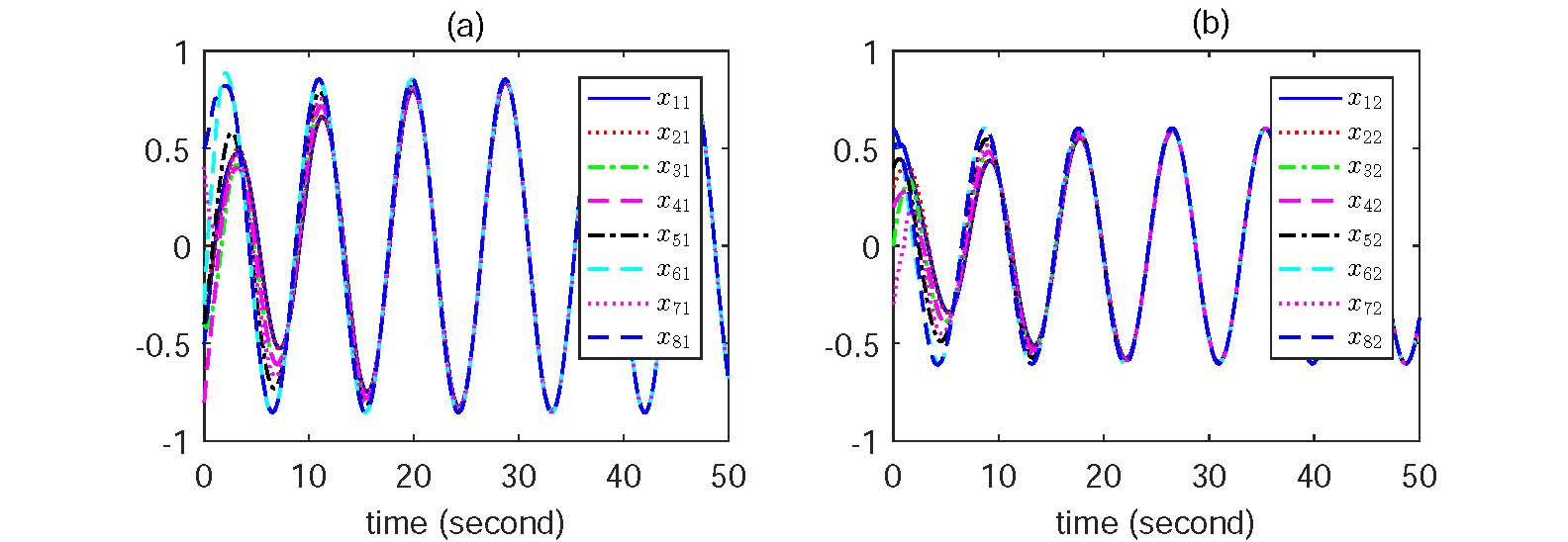}
\caption{Consensus of opinion dynamics model consisting of (\ref{sys}) and
(\ref{u4}) in the cooperative community.}%
\label{fig13}%
\end{figure}
\begin{figure}
\centering
\includegraphics[scale=0.65]{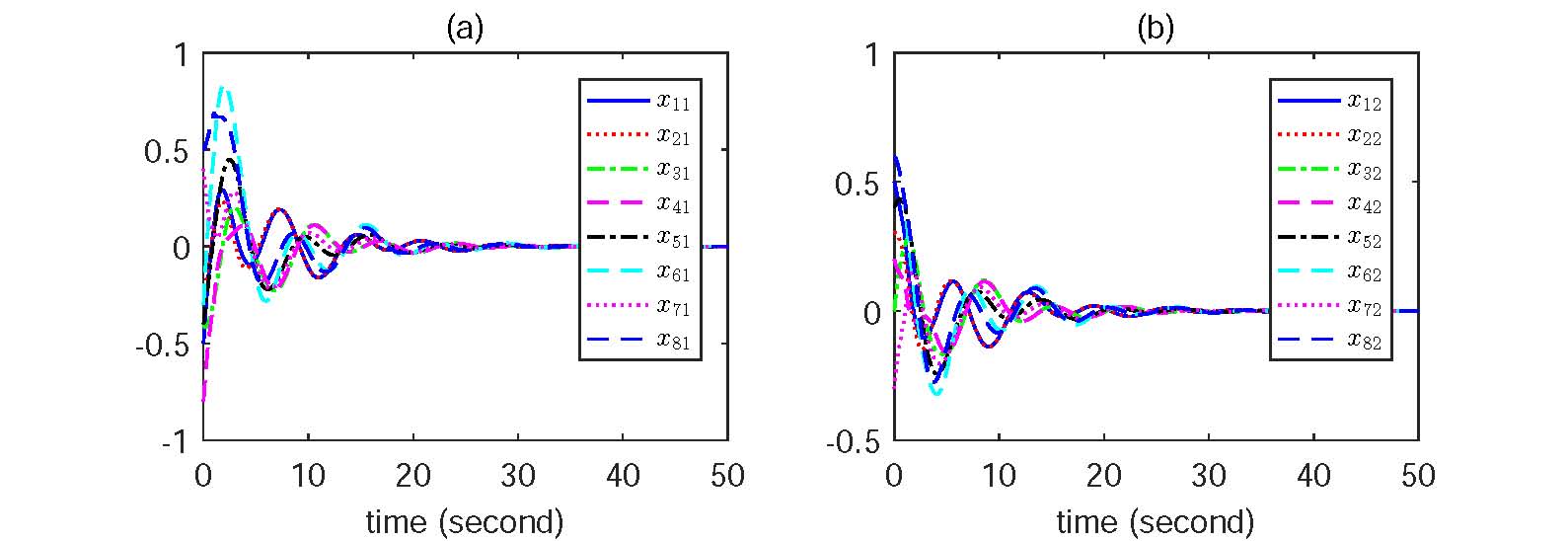}
\caption{Neutralization of opinion dynamics model consisting of (\ref{sys})
and (\ref{u4}).}%
\label{fig14}%
\end{figure}

\section{Simulation Analysis}\label{sec6}

In this section, two numerical examples are worked out to support the
obtained theoretical results and analyse the influences of the memory capacity. Consider a paradigmatic
social community consisting of eight individuals, which is shown in Fig. \ref{fig2}.

{\bf Example 1}: This example demonstrates the evolution of the individual
opinions on different topics under the proposed memoryless communication rule (\ref{u4}). Matrix
$A$ is given by
\[
A=\left[
\begin{array}
[c]{cc}%
0 & 1\\
-0.5 & 0
\end{array}
\right]  .
\]
Let $h=1$ and $\sigma=0.6$. Then, we have%
\[
\Lambda=\left[
\begin{array}
[c]{cc}%
1 & -1\\
0.5 & 1
\end{array}
\right]  ,\mathcal{A}=\left[
\begin{array}
[c]{cc}%
0.7204 & -1.2716\\
0.6358 & 0.7204
\end{array}
\right]  .
\]
We choose $\varrho=1/0.2>\max_{i\in \mathbf{I}[2,N]}\{1/\operatorname{Re}\{
\lambda_{i}\} \}$, by solving the parametric ARE (\ref{eq3}), and with the
help of (\ref{u4}), we have%
\[
F=(1/0.2)\left[
\begin{array}
[c]{cc}%
-0.0688 & -0.1336\\
0.1224 & -0.1525
\end{array}
\right],
\]
where $\gamma=0.2$, in which $\gamma^*$ can be estimated by Remark \ref{remark3}.

\begin{figure}
\centering
\includegraphics[scale=0.75]{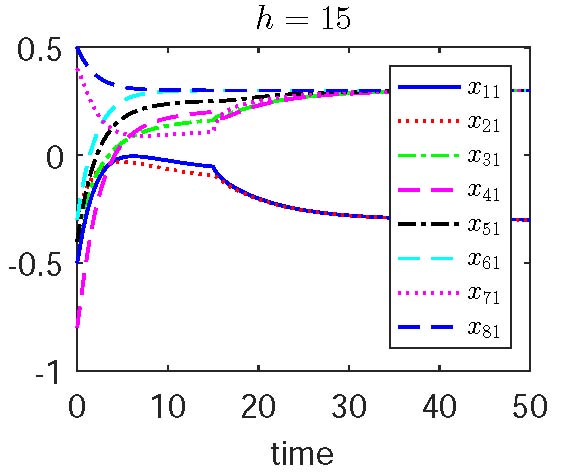}\includegraphics[scale=0.75]{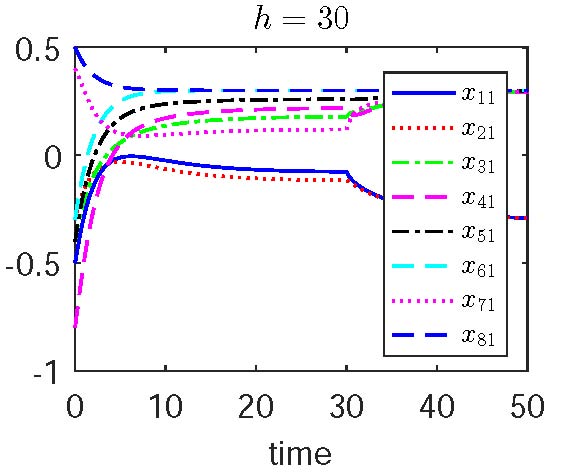}
\caption{Opinions formation processes are affected by $h$ (Equal
polarization).}%
\label{fig9}%
\end{figure}
\begin{figure}
\centering
\includegraphics[scale=0.75]{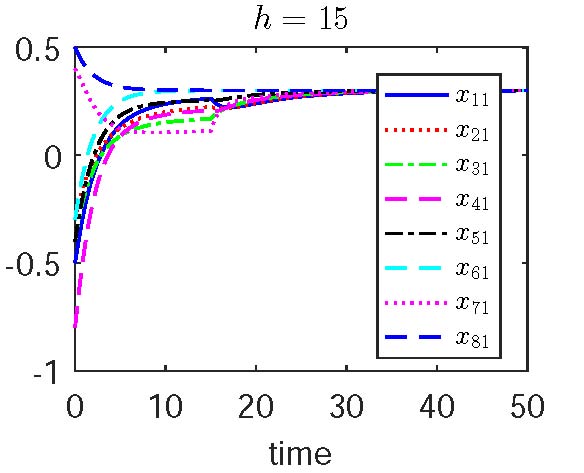}\includegraphics[scale=0.75]{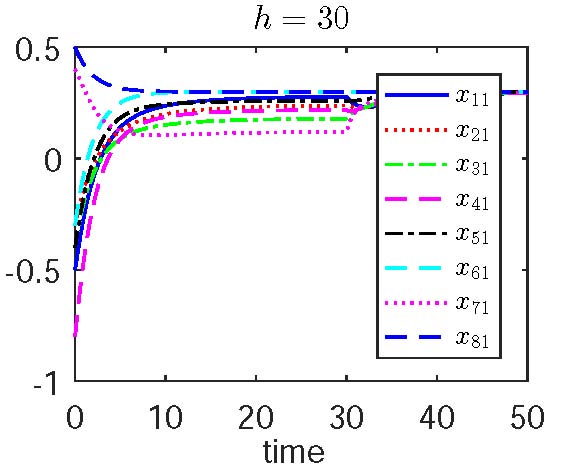}
\caption{Consensus formation processes are affected by $h$.}%
\label{fig10}%
\end{figure}
\begin{figure}
\centering
\includegraphics[scale=0.75]{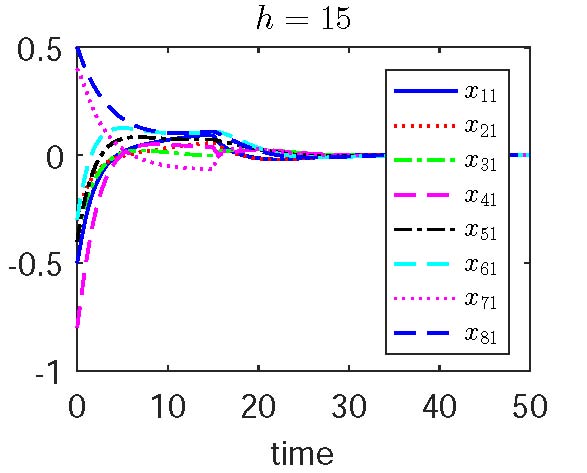}\includegraphics[scale=0.75]{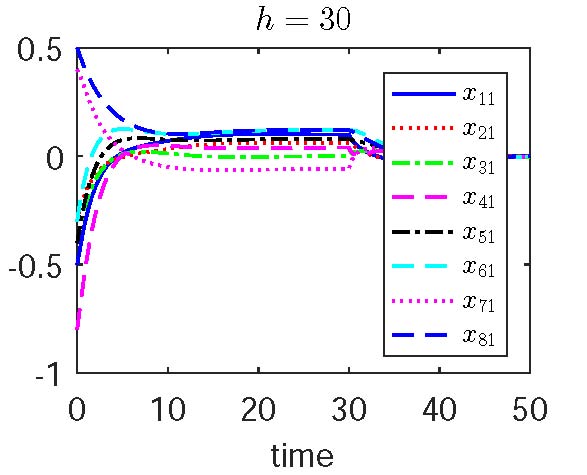}
\caption{Opinions formation processes are affected by $h$ (Neutralization).}%
\label{fig11}%
\end{figure}

{\bf Example 2}: This example studies the evolution of opinions
affected by the memory capacity $h$ under the proposed memory communication rule (\ref{u1}). For convenience of
description, we only consider the influence of the memory capacity on topic
$1$. Let $h=15$ (half a month) and $h=30$ (a month). The corresponding
matrices $\mathcal{A}$ are respectively computed as%
\begin{equation*}
\mathcal{A}=\left[
\begin{array}
[c]{cc}%
1.0857 & -0.5424\\
0 & 0.0010
\end{array}
\right]  \times10^{3},
\end{equation*}
and
\begin{equation*}
\mathcal{A}=\left[
\begin{array}
[c]{cc}%
1.9614 & -0.9807\\
0 & 0
\end{array}
\right]  \times10^{6}.
\end{equation*}
Let the initial opinions be the same as that in Section \ref{sec5}. The process of opinion evolution are shown in Fig. \ref{fig9}, Fig.
\ref{fig10} and Fig. \ref{fig11}, by which, it is clear to see that the memory capacity $h$ of the individuals
is inversely proportional to the speeds of the ultimate opinions formation. Moreover, we can see that, if the memory capacity $h$ is contained in both the opinion update mechanism and the communication rule (process), then the opinions is stationary. If the memory capacity $h$ is contained only in the opinion update mechanism, then the opinions is nonstationary (see Fig. \ref{fig15}, Fig. \ref{fig13} and Fig. \ref{fig14} for details).

In addition, although the coopetitive social network
$\mathcal{G}$ is structurally
balanced and quasi-strongly connected. Under the memory communication rules, it is interesting to see from Fig. \ref{fig7} that the individuals' opinions are neutralizable if $x_{8}\left(  0\right)=\left[0.5,0.6\right]  ^{\mathrm{T}}$ is replaced by $x_{8}\left(  0\right)  =\left[0.5,-0.6\right]  ^{\mathrm{T}}$, where $h=1$.
\begin{figure}
\centering
\includegraphics[scale=0.65]{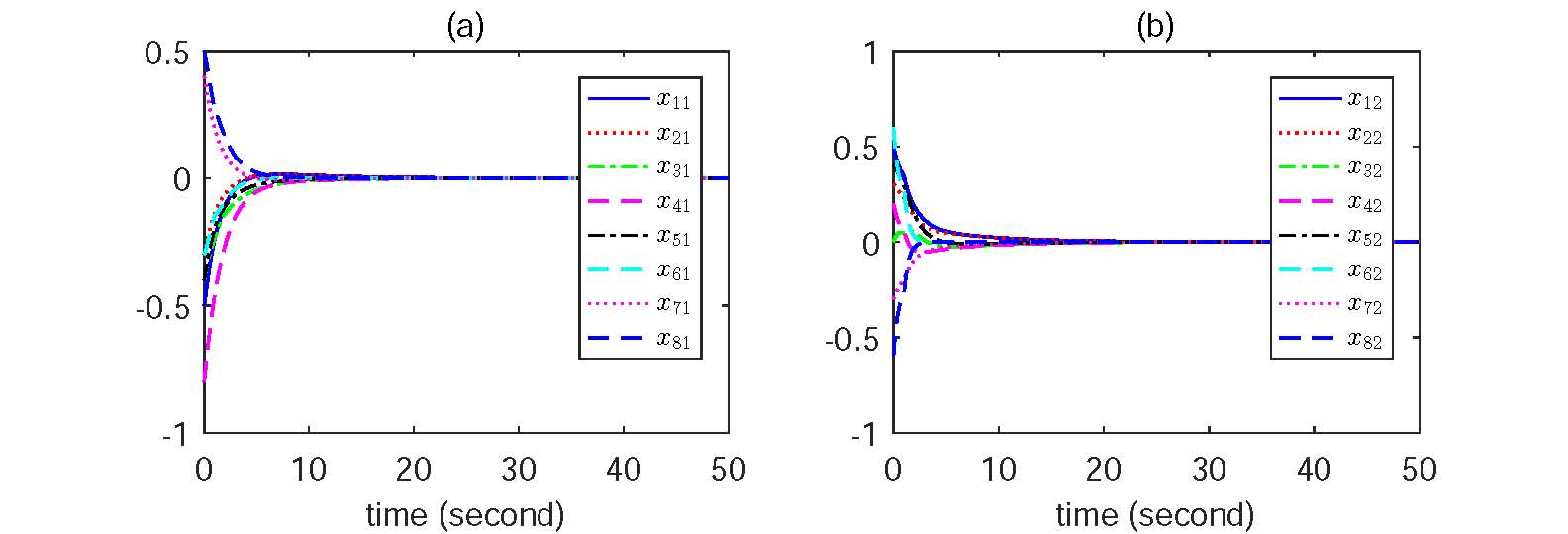}
\caption{Neutralization of opinion dynamics model consisting of (\ref{sys}) and (\ref{u1}).}%
\label{fig7}%
\end{figure}

\section{Conclusion}\label{sec7}

This paper proposed a novel model of opinion dynamics on coopetitive social
networks to describe the evolution of the individuals' opinions depending on its
own and neighbors' current opinions and past opinions, wherein the cooperative
and competitive interactions coexist. Memory and memoryless
communication rules have been given to analysis the evolution of opinions in
the coopetitive social networks, respectively. Sufficient and/or necessary
conditions guaranteeing the equal polarization, consensus and
neutralizability of the opinions were obtained on the basis of the network
topological structure and the spectral analysis. We use the proposed model to revisit Kahneman's seminal experiments on choices in risky and riskless contexts, and to reveal the insightful social
interpretations. By the simulation analysis,
it was shown that the memory capacity of the individuals was inversely
proportional to the speeds of the ultimate opinions formational. In addition, the opinions of all the individuals on different
topics were stationary under the memory communication rule and the opinions of
all the individuals on different topics were nonstationary under the
memoryless communication rule.

In the future works, we will consider the number of the individuals coming from a camp by the corresponding Laplacian matrix.
Moreover, the heterogeneous opinion dynamics model is another interesting topic and a system of more nodes will be discussed in our future study.
\section*{Acknowledgments}
The authors are very grateful
to the Associate Editor and the anonymous reviewers for their
comments which have helped to improve the quality of the paper
a lot. This work is supported by National
Natural Science Foundation of China under grant numbers 61903282 and 61625305,
and by China Postdoctoral Science Foundation funded project under grant number
2020T130488.

\section*{Appendix}

{\bf A1. Proof of Theorem \ref{theorem1}}

According to (\ref{x1}) and (\ref{v1}), (\ref{u2}) can be represented as
\begin{align}
v_{i}\left(  t\right)  =&-F\sum\limits_{j=1}^{N}l_{ij}\bigg(  \chi
_{j}\left(  t\right) +(1-\sigma)\nonumber \\
&+\int_{t-h}^{t}\mathrm{e}^{A\left(
t-h-s\right)  }\Lambda v_{j}\left(  s\right)  \mathrm{d}s\bigg),i\in \mathbf{I}[1,N].
\label{v2}%
\end{align}
Let $\varkappa$ and $\mu$ be defined as%
\[
\varkappa \triangleq \left[  \varkappa_{1}^{\mathrm{T}},\varkappa_{2}%
^{\mathrm{T}},\cdots,\varkappa_{N}^{\mathrm{T}}\right]  ^{\mathrm{T}%
}\triangleq \left(  \left(  DU\right)  ^{-1}\otimes I_{n}\right)  \chi,
\]
and%
\[
\mu \triangleq \left[  \mu_{1}^{\mathrm{T}},\mu_{2}^{\mathrm{T}},\cdots,\mu
_{N}^{\mathrm{T}}\right]  ^{\mathrm{T}}\triangleq \left(  \left(  DU\right)
^{-1}\otimes I_{n}\right)  v,
\]
where $\chi \triangleq \left[  \chi_{1}^{\mathrm{T}},\chi_{2}^{\mathrm{T}%
},\cdots,\chi_{N}^{\mathrm{T}}\right]  ^{\mathrm{T}}$, $v\triangleq \left[
v_{1}^{\mathrm{T}},v_{2}^{\mathrm{T}},\cdots,v_{N}^{\mathrm{T}}\right]
^{\mathrm{T}}$ and $\otimes$\ denotes the Kronecker product. The
signed graph $\mathcal{G}$ is structurally balanced and quasi-strongly
connected. Hence, by Lemma \ref{lemma2}, (\ref{sys}) and (\ref{v2}) can be
rewritten as, respectively,%
\begin{align*}
\dot{\varkappa}\left(  t\right)  =&\left(  I_{N}\otimes A\right)
\varkappa \left(  t\right)  +\left(  \sigma I_{N}\otimes \Lambda\right)  \mu \left(
t\right)  \\&+\left(  (1-\sigma)I_{N}\otimes \Lambda\right)  \mu \left(  t-h\right)  ,
\end{align*}
and%
\begin{align*}
\mu \left(  t\right)  =&-\left(  J_{\overline{L}}\otimes F\right)  \varkappa
\left(  t\right)  \\&-\left( (1-\sigma)J_{\overline{L}}\otimes F\int_{t-h}%
^{t}\mathrm{e}^{A\left(  t-h-s\right)  }\Lambda\mu \left(  s\right)  \mathrm{d}%
s\right).
\end{align*}
It follows from (\ref{j}) that $\dot{\varkappa}_{1}\left(  t\right)  = A\varkappa_{1}\left(  t\right)$ and
\begin{equation}
\left \{
\begin{array}
[c]{rl}%
\dot{\alpha}\left(  t\right)  = & \left(  I_{N-1}\otimes A\right)
\alpha \left(  t\right)  +\left(  I_{N-1}\otimes \sigma \Lambda \right)
\nu \left(  t\right)  \\
& +\left(  I_{N-1}\otimes(1-\sigma)\Lambda \right)  \nu \left(  t-h\right)  ,\\
\nu \left(  t\right)  = & -\left(  J\otimes F\right)  \varsigma \left(  t\right)
,
\end{array}
\right.  \label{sys3}%
\end{equation}
where $\alpha \triangleq \left[  \varkappa_{2}^{\mathrm{T}},\varkappa
_{3}^{\mathrm{T}},\cdots,\varkappa_{N}^{\mathrm{T}}\right]  ^{\mathrm{T}}$,
$\nu \triangleq \left[  \mu_{2}^{\mathrm{T}},\mu_{3}^{\mathrm{T}},\cdots,\mu
_{N}^{\mathrm{T}}\right]  ^{\mathrm{T}}$ and $$\varsigma \left(  t\right)
=\alpha \left(  t\right)  +(1-\sigma)\int_{-h}^{0}\mathrm{e}^{A\left(
-h-\theta \right)  }\Lambda \nu \left(  t+\theta \right)  \mathrm{d}\theta.$$ Let $\mathcal{L}\left(  f\left(  t\right)  \right)  $ denote the Laplace
transformation of the time function $f\left(  t\right)  $. Then,%
\begin{align*}
\mathcal{L}\left(  \varsigma \left(  t\right)  \right)   &  =X\left(  s\right)
+\left(  (1-\sigma)\int_{-h}^{0}\mathrm{e}^{A\left(  -h-\theta \right)
}\mathrm{e}^{s\theta}\Lambda \mathrm{d}\theta \right)  V\left(  s\right)  \\
&  =X\left(  s\right)  +\varrho V\left(  s\right)  ,
\end{align*}
where $\mathcal{L}\left(  \alpha \left(  t\right)  \right)  =X\left(  s\right)
$, $\mathcal{L}\left(  \nu \left(  t\right)  \right)  =V\left(  s\right)  $ and
$\varrho=(1-\sigma)\left(  sI_{n}-A\right)  ^{-1}\left(  \mathrm{e}%
^{-Ah}-\mathrm{e}^{-sh}\right)  \Lambda$.
Therefore, (\ref{sys3}) can be written in the frequency domain as%
\begin{equation}
\left[
\begin{array}
[c]{cc}%
\Omega_{0} & \Omega_{1}\\
 J\otimes F  & \Omega_{2}%
\end{array}
\right]  \left[
\begin{array}
[c]{c}%
X\left(  s\right)  \\
V\left(  s\right)
\end{array}
\right]  =0,\label{eq1}%
\end{equation}
where%
\[
\left \{
\begin{array}
[c]{rl}%
\Omega_{0}= & I_{N-1}\otimes \left(  sI_{n}-A\right)  ,\\
\Omega_{1}= & -I_{N-1}\otimes \left(  \sigma \Lambda+(1-\sigma)\Lambda
\mathrm{e}^{-sh}\right)  ,\\
\Omega_{2}= & I_{N-1}\otimes I_{n}+\left(  J\otimes F\varrho \right)  .
\end{array}
\right.
\]
It yields that the characteristic equation of (\ref{eq1}) is%
\[
\Delta \left(  s\right)  =\det \left[
\begin{array}
[c]{cc}%
\Omega_{0} & \Omega_{1}\\
J\otimes F& \Omega_{2}%
\end{array}
\right]  =0.
\]
Notice that%
\[
J\otimes F\left(  sI_{n}-A\right)  ^{-1}\Omega_{1}+\Omega_{2}=\Omega_{3},
\]
where $\Omega_{3}=I_{N-1}\otimes I_{n}+\left(  J\otimes F\varrho_{0}\left(
\sigma+(1-\sigma)\mathrm{e}^{-Ah}\right)  \Lambda \right)  $, in which
$\varrho_{0}=\left(  sI_{n}-A\right)  ^{-1}$. It
follows that%
\begin{align*}
\Delta \left(  s\right)   &  =\det \left(  \left[
\begin{array}
[c]{cc}%
I_{N-1}\otimes I_{n} & 0\\
-\left(J\otimes F\varrho_{0}\right) & I_{N-1}\otimes I_{n}%
\end{array}
\right]  \left[
\begin{array}
[c]{cc}%
\Omega_{0} & \Omega_{1}\\
  J\otimes F   & \Omega_{2}%
\end{array}
\right]  \right)  \\
&  =\det \left[
\begin{array}
[c]{cc}%
\Omega_{0} & \Omega_{1}\\
0 & \Omega_{3}%
\end{array}
\right]  \\
&  =\det \left(  \Omega_{0}\right)  \det \left(  \Omega_{3}\right)  \\
&  =\det \left(  \Omega_{0}\left(  I_{N-1}\otimes I_{n}+\left(  J\otimes
\varrho_{0}\mathcal{A} F\right)  \right)  \right)  \\
&  =\det \left(  I_{N-1}\otimes \left(  sI_{n}-A\right)+J\otimes
\mathcal{A}F\right)  \\
&  =\prod_{i=2}^{N}\left \vert sI_{n}-\left(A-\lambda_{i}\mathcal{A}F\right)\right \vert ,
\end{align*}
where we have used $\det(I_n+CE)=\det(I_n+EC)$ with $C$ and $E$ having appropriate dimensions.
In view of $A+\lambda_{i}\mathcal{A}F$ are Hurwitz, then we have%
\begin{equation}
\lim_{t\rightarrow \infty}\varkappa_{i}\left(  t\right)  =0,i\in \mathbf{I}%
[2,N].\label{xx1}%
\end{equation}
It follows from Lemma \ref{lemma2}, (\ref{j}) and (\ref{x1}) that
\begin{equation}
\varkappa \left(  t\right)  =J_{\overline{L}}\varphi \left(  t\right)  ,
\label{xx2}%
\end{equation}
where $\varphi \triangleq \left[  \varphi_{1}^{\mathrm{T}},\varphi
_{2}^{\mathrm{T}},\cdots,\varphi_{N}^{\mathrm{T}}\right]  ^{\mathrm{T}%
}\triangleq \left(  \left(  DU\right)  ^{-1}\otimes I_{n}\right)  x$, in which,
$x\triangleq \left[  x_{1}^{\mathrm{T}},x_{2}^{\mathrm{T}},\cdots
,x_{N}^{\mathrm{T}}\right]  ^{\mathrm{T}}$. It yields from (\ref{j}) and
(\ref{xx2}) that $\varkappa_{1}\left(  t\right)  =0$ and%
\[
\varkappa_{i}\left(  t\right)  =\lambda_{i}\varphi_{i}\left(  t\right)
+\epsilon_{i}\varphi_{i+1}\left(  t\right)  ,i\in \mathbf{I}[2,N],
\]
where $\varphi_{N+1}\left(  t\right)  =0$. By (\ref{xx1}), we can further get%
\[
\lim_{t\rightarrow \infty}\varphi_{i}\left(  t\right)  =0,i\in \mathbf{I}[2,N].
\]
Since the first column of $U$ is $\mathbf{1}_{N}$ and $x\left(  t\right)
=\left(  DU\otimes I_{n}\right)  \varphi \left(  t\right)  =\left(  U\otimes
I_{n}\right)  \left(  D\otimes I_{n}\right)  \varphi \left(  t\right)  $. Then,
\[
\lim_{t\rightarrow \infty}\left \Vert x_{i}\left(  t\right)  -d_i\varphi
_{1}\left(  t\right)  \right \Vert =0,i\in \mathbf{I}[1,N],
\]which means that the individuals' opinions are polarizable. We point out that
the condition $A\in \mathcal{S}$ is such that $\lim_{t\rightarrow \infty}%
x_{i}\left(  t\right)  =c$ for nonzero initial opinions, where $\left \vert
c\right \vert \neq0$ is a bounded real number.

{\bf A2. Proof of Theorem \ref{theorem3}}

By virtue of the proof of Theorem \ref{theorem1}, the opinion dynamics network
consisting of dynamic model (\ref{sys}) and the memory communication rule
(\ref{u1}) can be written as%
\begin{align}
\dot{x}\left(  t\right)  =&\big(  I_{N}\otimes A-\sigma L\otimes\Lambda F \nonumber\\&-\left(
1-\sigma \right)  L\otimes \mathrm{e}^{-Ah}\Lambda F\big)  x\left(  t\right)  ,
\label{sys4}%
\end{align}
by which, it is clear to see that the individuals' opinions in the coopetitive
social networks are neutralizable if and only if%
\begin{align*}
&\operatorname{Re} \left\{ \left(  I_{N}\otimes A-\sigma L\otimes \Lambda F-\left(
1-\sigma \right)  L\otimes \mathrm{e}^{-Ah}\Lambda F\right) \right \}\\=&\operatorname{Re}\{
\left(  A-\varepsilon \mathcal{A}F\right)\}<0,
\end{align*}
where $\varepsilon \in \lambda \left(  L\right)  $, in which, $\lambda \left(
L\right)  $ denotes the eigenvalue of $L$.

Sufficiency: Since $\mathcal{G}$ does not involve an in-isolated structurally
balanced subgraph, it follows that $0\notin \lambda \left(  L\right)  $. In
addition, notice that $\left(  A,\mathcal{A}\right)  $ is
controllable, then there exists a opinion adjustment matrix $F$ such that
$A-\varepsilon \mathcal{A}F,\forall \varepsilon \in \lambda \left(  L\right)  $ is
Hurwitz. Therefore, the individuals' opinions are neutralizable.

Necessity: If individuals' opinions are neutralizable, then we have
$\operatorname{Re}\{A-\varepsilon \mathcal{A}F\}<0,\forall \varepsilon \in
\lambda \left(  L\right)  $. We assume that $\mathcal{G}$\textit{ }is
structurally balanced or involves an in-isolated structurally balanced
subgraph, which leads to $0\in \lambda \left(  L\right)  $. However,
$\operatorname{Re}\{A-\varepsilon \mathcal{A}F\}<0$ for $\varepsilon=0$ implies
$A$ is a Hurwitz matrix, which is contradictory to the fact that $A$ is not
Hurwitz. Hence, we derive the necessity that $\mathcal{G}$ does not involve an
in-isolated structurally balanced subgraph. By adopting the same method
presented in \cite{mz10tac} to display that $\left(  A,\mathcal{A}\right)  $ is
controllable.

{\bf A3. A Useful Lemma}
\begin{lemma}
\label{lemma3}Let $F=-\mathcal{A}^{\mathrm{T}}P\left(  \gamma \right)  $ and
$P\left(  \gamma \right)  $ be unique positive definite solution to the
parametric ARE
\begin{equation}
A^{\mathrm{T}}P+PA-P\mathcal{A}\mathcal{A}^{\mathrm{T}}P=-\gamma P.
\label{addeq3}%
\end{equation}
Then there exists a scalar $\gamma^{\ast}>0$ such that the following
closed-loop system%
\begin{equation}
\left \{
\begin{array}
[c]{rl}%
\dot{\eta}\left(  t\right)  = & A\eta \left(  t\right)  +\sigma \Lambda u\left(
t\right)  +(1-\sigma)\Lambda u\left(  t-h\right)  ,\\
u\left(  t\right)  = & \mathcal{F}\eta \left(  t\right)  ,
\end{array}
\right.  \label{delaysystem}%
\end{equation}
is asymptotically stable for all $\gamma \in \left(  0,\gamma^{\ast}\right)  $,
where $\mathcal{F}=\lambda \varrho F$, $\varrho \geq1/\operatorname{Re}\{
\lambda \}$ and $\operatorname{Re}\{ \lambda \}>0$.
\end{lemma}

\begin{proof}
Let
\begin{align}
z\left(  t\right)   &  =\eta \left(  t\right)  +(1-\sigma)\int_{t-h}%
^{t}\mathrm{e}^{A\left(  t-h-s\right)  }\Lambda u\left(  s\right)  \mathrm{d}%
s\nonumber \\
&  \triangleq \eta \left(  t\right)  +\psi \left(  t\right)  . \label{eq8}%
\end{align}
It then follows from (\ref{delaysystem}) and (\ref{eq8}) that%
\begin{equation}
\dot{z}\left(  t\right)  =Az\left(  t\right)  +\mathcal{A}u\left(  t\right)  ,
\label{eq9}%
\end{equation}
where $\mathcal{A}=\sigma\Lambda+\left(  1-\sigma\right)  \mathrm{e}^{-Ah}\Lambda$.
It yields from (\ref{eq9}),
(\ref{eq8}) and the second equation in (\ref{delaysystem}) that%
\[
\dot{z}\left(  t\right)  =\left(  A+\mathcal{AF}\right)  z\left(  t\right)
-\mathcal{AF}\psi \left(  t\right)  .
\]
We now consider the Lyapunov function $V_{1}\left(  z\left(  t\right)
\right)  =z^{\mathrm{T}}\left(  t\right)  Pz\left(  t\right)  $. Therefore,
for any given $\gamma \in \left(  0,\gamma_{0}\right)  $, where $\gamma_{0}>0$,
we get%
\begin{align}
 &\dot{V}_{1}\left(  z\left(  t\right)  \right) =z^{\mathrm{T}}\left(
t\right)  \left(  \left(  A+\mathcal{AF}\right)  ^{\mathrm{T}}P+P\left(
A+\mathcal{AF}\right)  \right)  z\left(  t\right)  \nonumber \\
&  -2z^{\mathrm{T}}\left(  t\right)  P\mathcal{AF}\psi \left(  t\right)
\nonumber \\
=& z^{\mathrm{T}}\left(  t\right) \big (P\mathcal{AA}^{\mathrm{T}}P-\lambda
\varrho P\mathcal{AA}^{\mathrm{T}}P-\gamma P -\lambda \varrho P\mathcal{AA}^{\mathrm{T}}P\big)z\left(  t\right)
+2\lambda \varrho z^{\mathrm{T}}\left(  t\right)  P\mathcal{AA}^{\mathrm{T}%
}P\psi \left(  t\right)  \nonumber \\
\leq & z^{\mathrm{T}}\left(  t\right)  \left(  -\left \Vert \lambda \right \Vert
\varrho P\mathcal{AA}^{\mathrm{T}}P-\gamma P\right)  z\left(  t\right)
 +\left \Vert \lambda \right \Vert \varrho z^{\mathrm{T}}\left(  t\right)  P\mathcal{AA}^{\mathrm{T}%
}Pz\left(  t\right)  +\left \Vert \lambda \right \Vert \varrho \psi^{\mathrm{T}}\left(  t\right)  P\mathcal{AA}%
^{\mathrm{T}}P\psi \left(  t\right)  \nonumber \\
\leq &  -\gamma z^{\mathrm{T}}\left(  t\right)  Pz\left(  t\right)
+\gamma \left \Vert \lambda \right \Vert \varrho n\psi^{\mathrm{T}}\left(
t\right)  P\psi \left(  t\right)  ,\label{eq10}%
\end{align}
where we have used $P\mathcal{AA}^{\mathrm{T}}P\leq n\gamma P$. By virtue of%
\begin{align*}
&z^{\mathrm{T}}\left(  t\right)  Pz\left(  t\right) \\ =&\eta^{\mathrm{T}}\left(
t\right)  P\eta \left(  t\right)  +2\eta^{\mathrm{T}}\left(  t\right)
P\psi \left(  t\right)  +\psi^{\mathrm{T}}\left(  t\right)  P\psi \left(
t\right)  ,
\end{align*}
and
\begin{align}
& \frac{1}{2}\gamma \eta^{\mathrm{T}}\left(  t\right)  P\eta \left(  t\right)
+2\gamma \psi^{\mathrm{T}}\left(  t\right)  P\psi \left(  t\right)  \nonumber \\
=&\left(  \sqrt{\frac{\gamma}{2}}P^{\frac{1}{2}}\eta \left(  t\right)
\right)  ^{\mathrm{T}}\left(  \sqrt{\frac{\gamma}{2}}P^{\frac{1}{2}}%
\eta \left(  t\right)  \right)   +\left(  \sqrt{2\gamma}P^{\frac{1}{2}}\psi \left(  t\right)  \right)
^{\mathrm{T}}\left(  \sqrt{2\gamma}P^{\frac{1}{2}}\psi \left(  t\right)
\right)  \nonumber \\
\geq&2\left(  \sqrt{\frac{\gamma}{2}}P^{\frac{1}{2}}\eta \left(  t\right)
\right)  ^{\mathrm{T}}\left(  \sqrt{2\gamma}P^{\frac{1}{2}}\psi \left(
t\right)  \right)  \nonumber \\
=& 2\gamma \eta^{\mathrm{T}}\left(  t\right)  P\psi \left(  t\right),
\label{addeq1}%
\end{align}
it follows from (\ref{eq10}) and (\ref{addeq1}) that
\begin{align}
\dot{V}_{1}\left(  z\left(  t\right)  \right)  \leq & -\gamma \eta
^{\mathrm{T}}\left(  t\right)  P\eta \left(  t\right)  -\gamma \psi^{\mathrm{T}%
}\left(  t\right)  P\psi \left(  t\right)   +\frac{1}{2}\gamma \eta^{\mathrm{T}}\left(  t\right)  P\eta \left(  t\right)
+2\gamma \psi^{\mathrm{T}}\left(  t\right)  P\psi \left(  t\right) \nonumber \\& +\gamma \left \Vert \lambda \right \Vert \varrho n\psi^{\mathrm{T}}\left(  t\right)  P\psi \left(
t\right)  \nonumber \\
=& -\frac{1}{2}\gamma \eta^{\mathrm{T}}\left(  t\right)  P\eta \left(
t\right)    +\left(  1+\left \Vert \lambda \right \Vert \varrho n\right)  \gamma \psi^{\mathrm{T}}\left(
t\right)  P\psi \left(  t\right)  .\label{eq12}%
\end{align}
According to Lemma \ref{lemmaadd}, we have%
\begin{align}
\psi^{\mathrm{T}}\left(  t\right)  P\psi \left(  t\right)=&
(1-\sigma)^{2}\left(  \int_{t-h}^{t}\mathrm{e}^{A\left(  t-h-s\right)
}\Lambda u\left(  s\right)  \mathrm{d}s\right)  ^{\mathrm{T}} P\left(  \int_{t-h}^{t}\mathrm{e}^{A\left(  t-h-s\right)  }\Lambda
u\left(  s\right)  \mathrm{d}s\right)  \nonumber \\
\leq &  (1-\sigma)^{2}h\int_{t-h}^{t}u^{\mathrm{T}}\left(  s\right)
\Lambda^{\mathrm{T}}\mathrm{e}^{A^{T}\left(  t-h-s\right)  } P\mathrm{e}^{A\left(  t-h-s\right)  }\Lambda u\left(  s\right)
\mathrm{d}s\nonumber \\
\leq & (1-\sigma)^{2}h\delta_{0}\int_{t-h}^{t}u^{\mathrm{T}}\left(
s\right)  u\left(  s\right)  \mathrm{d}s,\label{eq11}%
\end{align}
where $\delta_{0}=\left \Vert \Lambda\right \Vert ^{2}\cdot \left \Vert \mathrm{e}%
^{Ah}\right \Vert ^{2}\left \Vert P \right \Vert$. We note that%
\begin{align*}
u^{\mathrm{T}}\left(  s\right)  u\left(  s\right)   &  =\left( \lambda
\varrho \right)  ^{2}\eta^{\mathrm{T}}\left(  s\right)  P\mathcal{AA}%
^{\mathrm{T}}P\eta \left(  s\right) \\
&  \leq \left \Vert \lambda\right \Vert^{2}\varrho^{2}n\gamma \eta^{\mathrm{T}}\left(
s\right)  P\eta \left(  s\right)  .
\end{align*}
Hence, equation (\ref{eq11}) can be written as%
\begin{equation}
\psi^{\mathrm{T}}\left(  t\right)  P\psi \left(  t\right)  \leq \delta_{1}%
\gamma\int_{t-h}^{t}\eta^{\mathrm{T}}\left(  s\right)  P\eta \left(
s\right)  \mathrm{d}s, \label{eq13}%
\end{equation}
where $\delta_{1}=(1-\sigma)^{2}\left \Vert \lambda \right \Vert ^{2}\varrho
hn\delta_{0}$. Substituting (\ref{eq13}) into (\ref{eq12}) gives%
\[
\dot{V}_{1}\left(  z\left(  t\right)  \right)  \leq-\frac{1}{2}\gamma
\eta^{\mathrm{T}}\left(  t\right)  P\eta \left(  t\right)  +\delta_{2}\left(
\gamma \right)  \int_{t-h}^{t}\eta^{\mathrm{T}}\left(  s\right)  P\eta \left(
s\right)  \mathrm{d}s,
\]
where $\delta_{2}\left(  \gamma \right)  =\left(  1+\left \Vert \lambda
\right \Vert \varrho n\right)  \delta_{1}\gamma^{2}$. Now, we choose another
two Lyapunov functional as%
\[
V_{2}\left(  \eta_{t}\right)  =\delta_{2}\left(  \gamma \right)  \int_{0}%
^{h}\left(  \int_{t-s}^{t}\eta^{\mathrm{T}}\left(  l\right)  P\eta \left(
l\right)  \mathrm{d}l\right)  \mathrm{d}s,
\]
and%
\[
V_{3}\left(  \eta_{t}\right)  =2\delta_{1}\gamma^{2}\int_{t-h}^{t}%
\eta^{\mathrm{T}}\left(  s\right)  P\eta \left(  s\right)  \mathrm{d}s.
\]
It follows that%
\begin{align*}
\dot{V}\left(  \eta_{t}\right)  = &  \dot{V}_{1}\left(  z\left(  t\right)
\right)  +\dot{V}_{2}\left(  \eta_{t}\right)  +\dot{V}_{3}\left(  \eta
_{t}\right)  \\
\leq &  -\frac{1}{2}\gamma \eta^{\mathrm{T}}\left(  t\right)  P\eta \left(
t\right)  +\delta_{2}\left(  \gamma \right)  \int_{t-h}^{t}\eta^{\mathrm{T}%
}\left(  s\right)  P\eta \left(  s\right)  \mathrm{d}s +\delta_{2}\left(  \gamma \right)  h\eta^{\mathrm{T}}\left(  t\right)
P\eta \left(  t\right)  +2\delta_{1}\gamma^{2}\eta^{\mathrm{T}}\left(
t\right)  P\eta \left(  t\right)  \\
&  -\delta_{2}\left(  \gamma \right)  \int_{0}^{h}\eta^{\mathrm{T}}\left(
t-s\right)  P\eta \left(  t-s\right)  \mathrm{d}s -2\delta_{1}\gamma^{2}\eta^{\mathrm{T}}\left(  t-h\right)  P\eta \left(
t-h\right)  \\
\leq &  -\left(  \frac{1}{2}-h\left(  1+\left \Vert \lambda \right \Vert \varrho
n\right)  \delta_{1}\gamma-2\delta_{1}\gamma \right)  \gamma \eta^{\mathrm{T}}\left(  t\right)  P\eta \left(  t\right)  .
\end{align*}
Let $\gamma^{\ast}\in \left(  0,\gamma_{0}\right)  $ be such that%
\[
\frac{1}{2}-h\left(  1+\left \Vert \lambda \right \Vert \varrho n\right)
\delta_{1}\gamma-2\delta_{1}\gamma \geq \frac{1}{4},\gamma \in \left(
0,\gamma^{\ast}\right)  .
\]
It then yields that%
\[
\dot{V}\left(  \eta_{t}\right)  \leq-\frac{1}{4}\gamma \eta^{\mathrm{T}}\left(
t\right)  P\eta \left(  t\right)  \leq-\frac{1}{4}\delta \gamma^{1+a}\left \Vert
\eta \left(  t\right)  \right \Vert ^{2},
\]
where we have used there exists a constant $\delta>0$ and an integer $a\geq1$
such that $P\geq \delta \gamma^{a}I_{n}$. Clearly, there exists a constant
$\delta_{4}>0$ such that $V\left(  \eta_{t}\right)  \leq \delta_{4}\left \Vert
\eta \left(  t\right)  \right \Vert ^{2}$. In the following, we will show that
there exists a constant $\delta_{3}>0$ such that $V\left(  \eta_{t}\right)
\geq \delta_{3}\left \Vert \eta \left(  t\right)  \right \Vert ^{2}$. Notice that%
\begin{align*}
  & -2\eta^{\mathrm{T}}\left(  t\right)  P\psi \left(  t\right) \leq \frac
{1}{2}\eta^{\mathrm{T}}\left(  t\right)  P\eta \left(  t\right)  +2\psi
^{\mathrm{T}}\left(  t\right)  P\psi \left(  t\right) \\
&  \leq2\delta_{1}\gamma^{2}\int_{t-h}^{t}\eta^{\mathrm{T}}\left(  s\right)
P\eta \left(  s\right)  \mathrm{d}s+\frac{1}{2}\eta^{\mathrm{T}}\left(
t\right)  P\eta \left(  t\right)  ,
\end{align*}
by which, we can obtain%
\begin{align*}
V_{1}\left(  z\left(  t\right)  \right)   &  \geq \eta^{\mathrm{T}}\left(
t\right)  P\eta \left(  t\right)  +2\eta^{\mathrm{T}}\left(  t\right)
P\psi \left(  t\right) \\
&  \geq \frac{1}{2}\eta^{\mathrm{T}}\left(  t\right)  P\eta \left(  t\right)
-2\delta_{1}\gamma^{2}\int_{t-h}^{t}\eta^{\mathrm{T}}\left(  s\right)
P\eta \left(  s\right)  \mathrm{d}s\\
&  =\frac{1}{2}\eta^{\mathrm{T}}\left(  t\right)  P\eta \left(  t\right)
-V_{3}\left(  \eta_{t}\right)  .
\end{align*}
Then, we can further get%
\begin{align*}
V\left(  \eta_{t}\right)   &  \geq V_{1}\left(  z\left(  t\right)  \right)
+V_{3}\left(  \eta_{t}\right) \\
&  \geq \frac{1}{2}\delta \gamma^{a}\left \Vert \eta \left(  t\right)  \right \Vert
^{2}.
\end{align*}
Therefore, the closed-loop system (\ref{delaysystem})
is asymptotically stable for $\gamma \in \left(  0,\gamma^{\ast}\right)  $.
\end{proof}

{\bf A4. Proof of Theorem \ref{theorem2}}

Notice that the opinion dynamic network consisting of opinion dynamics model
(\ref{sys}) and the memoryless communication rule (\ref{u4}) is given by%
\begin{align}
\dot{x}_{i}\left(  t\right)  = & Ax_{i}\left(  t\right)  +\sigma%
\Lambda \digamma \chi_{i}\left(  t\right)  +(1-\sigma)\Lambda \digamma
\chi_{i}\left(  t-h\right)  \nonumber \\
=& Ax_{i}\left(  t\right)  +\sigma\Lambda \digamma \sum_{j\in
\mathcal{N}_{i}}l_{ij}x_{j}\left(  t\right)  +(1-\sigma)\nonumber \\
&\times \Lambda \digamma \sum_{j\in \mathcal{N}_{i}}l_{ij}x_{j}\left(  t-h\right)
,i \in \mathbf{I}[1,N],\label{eq4}%
\end{align}
where $\digamma=\varrho F$, in which $F=-\mathcal{A}^{\mathrm{T}}P\left(
\gamma \right)  $. Let $x=\left[  x_{1}^{\mathrm{T}},x_{2}^{\mathrm{T}}%
,\cdots,x_{N}^{\mathrm{T}}\right]  ^{\mathrm{T}}$, then (\ref{eq4}) can be expressed as the following compact form,
\begin{align}
\dot{x}\left(  t\right)  = &  \left(  I_{N}\otimes A\right)  x\left(
t\right)  +\left(  \sigma L\otimes \Lambda \digamma \right)  x\left(  t\right)
+\left(  (1-\sigma)L\otimes \Lambda \digamma \right)  x\left(  t-h\right)
\nonumber \\
= &  \left(  \left(  DU\right)  \otimes I_{n}\right)  \left(  I_{N}\otimes
A\right)  \left(  \left(  DU\right)  ^{-1}\otimes I_{n}\right)  x\left(
t\right)  +\left(  \left(  DU\right)  \otimes I_{n}\right)  \left(  \sigma
J_{\overline{L}}\otimes \Lambda \digamma \right)  \left(  \left(  DU\right)
^{-1}\otimes I_{n}\right)  x\left(  t\right)  \nonumber \\
&  +\left(  \left(  DU\right)  \otimes I_{n}\right)  \left(  (1-\sigma
)J_{\overline{L}}\otimes \Lambda \digamma \right)  \left(  \left(  DU\right)  ^{-1}\otimes I_{n}\right)  x\left(
t-h\right)  ,\label{eq5}%
\end{align}
where we have used Lemma \ref{lemma2}. Define a set of new variables
$\varsigma=\left(  \left(  DU\right)  ^{-1}\otimes I_{n}\right)  x=\left[
\varsigma_{1}^{\mathrm{T}},\varsigma_{2}^{\mathrm{T}},\cdots,\varsigma
_{N}^{\mathrm{T}}\right]  ^{\mathrm{T}}$, it follows from (\ref{j}) that (\ref{eq5}) is equivalent to $\dot{\varsigma}_{1}\left(
t\right)  =A\varsigma_{1}\left(  t\right)  $ and%
\begin{align*}
\dot{\varsigma}_{i}\left(  t\right)  =&A\varsigma_{i}\left(  t\right)
+\lambda_{i}\sigma\Lambda\digamma \varsigma_{i}\left(  t\right)  +\lambda
_{i}(1-\sigma)\Lambda\digamma \varsigma_{i}\left(  t-h\right) +\epsilon
_{i}\sigma\Lambda\digamma \varsigma_{i+1}\left(  t\right)+\epsilon_{i}%
(1-\sigma)\Lambda\digamma \varsigma_{i+1}\left(  t-h\right)  ,
\end{align*}
where $i\in \mathbf{I}[2,N]$ and $\varsigma_{N+1}\left(  t\right)  =0$. Since
$x=\left(  \left(  DU\right)  \otimes I_{n}\right)  \varsigma$, it yields from
the special structure of $U$ that the polarization is achieved if%
\begin{equation}
\lim_{t\rightarrow \infty}\left \Vert \varsigma_{i}\left(  t\right)  \right \Vert
=0,i\in \mathbf{I}[2,N], \label{eq7}%
\end{equation}
by which, we have%
\[
\lim_{t\rightarrow \infty}\left \Vert x_{i}\left(  t\right)  -d_i\varsigma
_{1}\left(  t\right)  \right \Vert =0,i\in \mathbf{I}[1,N].
\]
Clearly, (\ref{eq7}) is true if and only if%
\begin{equation}
\dot{\varsigma}_{i}\left(  t\right)  =A\varsigma_{i}\left(  t\right)
+\lambda_{i}\sigma\Lambda\digamma \varsigma_{i}\left(  t\right)  +\lambda
_{i}(1-\sigma)\Lambda\digamma \varsigma_{i}\left(  t-h\right)  , \label{eq6}%
\end{equation}
is asymptotically stable, $i\in\mathbf{I}[2,N]$. In what follows, we will show that system
(\ref{eq6}) is indeed asymptotically stable if $\varrho \geq \max
_{i\in \mathbf{I}[2,N]}\{1/\operatorname{Re}\{ \lambda_{i}\} \}$.
Notice that the stability of (\ref{eq6}) is equivalent to the stability
of the following system%
\[
\left \{
\begin{array}
[c]{rl}%
\dot{\eta}\left(  t\right)  = & A\eta \left(  t\right)  +\sigma \Lambda u\left(
t\right)  +(1-\sigma)\Lambda u\left(  t-h\right)  ,\\
u\left(  t\right)  = & \mathcal{F}\eta \left(  t\right)  ,
\end{array}
\right.
\]
where $\mathcal{F}=\lambda \varrho F$, $\varrho \geq1/\operatorname{Re}\{
\lambda \}$ and $\operatorname{Re}\{ \lambda \}>0$. The rest of proof is completed by Lemma \ref{lemma3}.

\batchmode
\frenchspacing
\bibliographystyle{amsplain}
\bibliography{LiuChai}
\batchmode
\end{document}